%% file: mismatchedBC_arxiv.tex
\documentclass[10pt, twocolumn]{IEEEtran}
\input{preamble}

\graphicspath{{./figures/}}
\usepackage{graphicx,cite}
\usepackage{epstopdf}
\usepackage{enumerate}
\usepackage[ colorlinks = true,
 linkcolor = blue,
 urlcolor  = blue,
 citecolor = red,
 anchorcolor = green,]{hyperref}
\usepackage{soul,color}

\newcommand{\myfoot}[1]{\footnote{\color{red}\bf #1}}
\soulregister\em7
\soulregister\cite7
\soulregister\ref7
\soulregister\eqref7
\soulregister\underline7
\soulregister\emph7
\soulregister\footnote7
\soulregister\myfoot7
\definecolor{Dyellow}{RGB}{254,152,0}
\definecolor{Dgreen}{RGB}{0,176,80}

\begin{document}

\title{The Dispersion of Broadcast Channels With Degraded Message Sets Using Gaussian Codebooks}

%\author{
%\IEEEauthorblockN{Zhuangfei Wu, Lin Zhou, Jinpeng Xu and Lin Bai}
%\IEEEauthorblockA{School of Cyber Science and Technology (CST), Beihang University\\
%Email: \{{zhuangfeiwu}, {lzhou}, {xujinpeng}, {l.bai}\}@buaa.edu.cn}
%\red{Comment: change this part}
%}
\author{Zhuangfei Wu, Lin Bai, Jinpeng Xu, Lin Zhou and Mehul Motani
\thanks{This work has been partially presented at IEEE ICCC 2024~\cite{wu2024iccc}.}
\thanks{Z. Wu,  L. Bai, J. Xu and L. Zhou are with the School of Cyber Science and Technology, Beihang University, Beijing, China, 100191 (Emails: \{zhuangfeiwu, l.bai, xujinpeng, lzhou\}@buaa.edu.cn). }
\thanks{M. Motani is with the Department of Electrical and Computer
Engineering, National University of Singapore, Singapore, 117583 (Email:
motani@nus.edu.sg).}}

\maketitle
\flushbottom
\allowdisplaybreaks[1]

\begin{abstract}
We study the two-user broadcast channel with degraded message sets and derive second-order achievability rate regions. Specifically, the channel noises are not necessarily Gaussian and we use spherical codebooks for both users. The weak user with worse channel quality applies nearest neighbor decoding by treating the signal of the other user as interference. For the strong user with better channel quality, we consider two decoding schemes: successive interference cancellation (SIC) decoding and joint nearest neighbor (JNN) decoding. We adopt two performance criteria: separate error probabilities (SEP) and joint error probability (JEP). Counterintuitively, the second-order achievable rate regions under both SIC and JNN decoding are identical although JNN decoding usually yields better performance in other multiterminal problems with Gaussian noise. Furthermore, we generalize our results to the case with quasi-static fading and show that the asymptotic notion of outage capacity region is an accurate performance measure even at finite blocklengths.
\end{abstract}


\section{Introduction}\label{sec:intro}
With growing research interest in 6G networks and internet of things, low-latency communication is desired in order to serve delay-sensitive services~\cite{letaief2019roadmap,letaief2021edge}. Compared with  asymptotic analysis where the blocklength and the communication delay tend to infinity~\cite{shannon1948mathematical}, finite-blocklength analysis aims to characterize the theoretical performance of a communication system with low-latency. In particular, Polyanskiy, Poor and Verdú~\cite{polyanskiy2010finite} established the second-order asymptotics for point-to-point channel coding, which generalizes the seminal analyses of Strassen and Hayashi~\cite{hayashi2009information,strassen2009asymptotic}. The dominant backoff at the finite blocklength from the first-order asymptotic channel capacity is called the second-order term, and the coefficient for this term is known as \emph{channel dispersion}.

The above studies have promoted the development of finite blocklength analysis for network information theory\cite{TanBook,zhou2023monograph}. As a ubiquitous multi-terminal communication model, the broadcast channel (BC) pioneered by Cover~\cite{cover1972broadcast} consists of one transmitter and multiple receivers, where the transmitter aims to send a message to each receiver. Following Cover's work, Bergmans~\cite{bergmans1973broadcast,bergmans1974broadcast} derived the capacity region of the degraded BC where degraded means the channel quality of one user is always better than that of the other. Subsequently, Kaspi-Merhav~\cite{kaspi2011exponents} and Weng-Pradhan-Anastasopoulos~\cite{weng2008exponents} derived bounds on error exponents (a.k.a., the exponential decay rates of error probabilities) for BC with degraded message sets, also known as asymmetric BC (ABC), where the transmitter sends a common message to both receivers and sends an additional private message to only one of the receivers. Recently, Tan and Kosut~\cite{tan2014dispersions} characterized the second-order achievable rate region for the discrete memoryless ABC.

Although the above studies have comprehensively studied the ABC, they share one common limitation: the distribution of channel noise is assumed to be known in order to design and optimize the code. However, it is usually challenging to obtain the exact distribution of channel noise in practical systems. Instead, it is feasible to measure the statistics, i.e., the moments of the noise. Thus, one could design coding schemes with moments constraints. This direction was pioneered by Lapidoth~\cite{lapidoth1996} who proposed to use a nearest neighbor (NN) decoder with i.i.d. Gaussian codebooks or spherical codebooks for a memoryless channel with additive arbitrary noise. Such a coding scheme is known as \emph{mismatched} coding since the proposed coding scheme is not necessarily optimal for the additive noise distribution. Lapidoth derived the maximal achievable rate and showed that such a coding scheme is asymptotically optimal when specialized to the additive white Gaussian noise (AWGN) channel. Recently, Scarlett, Tan and Durisi~\cite{scarlett2017mismatch} refined Lapidoth's result by establishing the ensemble tight second-order asymptotic rate and also considering the  interference channel. In particular, the authors of~\cite{scarlett2017mismatch} showed that the spherical codebook has a strictly better second-order performance than the i.i.d. Gaussian codebook although the asymptotic achievable rates of both codebooks are identical.

However, the study of ABC with additive non-Gaussian noise remains unexplored. It is worthwhile to generalize~\cite{lapidoth1996,tan2014dispersions} to the ABC setting since ABC is a building block for downlink multi-user wireless communications. To do so, we propose a mismatched coding scheme and derive achievable second-order rate regions under both separate error probabilities (SEP) and joint error probability (JEP) criteria when either successive interference cancellation (SIC) decoding and joint nearest neighbor (JNN) decoding is used by the strong user. Furthermore, we generalize our results to the case with quasi-static fading~\cite[Section 5.4.1]{tse2005fundamentals}, where the random fading coefficients remain unchanged during the transmission of each codeword. As argued by~\cite{yang2014quasi}, quasi-static fading is valid for low-latency communication with  finite blocklength.

\subsection{Main Contributions}
We propose a coding scheme using spherical codebooks and a NN decoder. The weak user with worse channel quality applies nearest neighbor decoding by treating the signal of the other user as interference. For the strong user with better channel quality, we consider two decoding schemes: SIC and JNN decoding. We consider two types of performance criteria: SEP and JEP. Under the above setting, we derive achievable second-order rate regions.

Under SEP, we generalize the second-order asymptotic analyses of point-to-point mismatched coding~\cite{scarlett2017mismatch} to the ABC.  Counter-intuitively, the achievable second-order rate regions under both SIC and JNN decoding are identical although JNN decoding usually yields better performance in other multiterminal problems with Gaussian noise~\cite[Theorem 2]{yavas2021gaussian},\cite[Theorem 20]{zhou2016second}. Our achievability results also hold for a symmetric broadcast channel (SBC) with non-Gaussian noise, where each user decodes one private message and thus the message sets are symmetric. Under JEP, we establish similar results where SIC and JNN decoding have same second-order performance. Furthermore, numerically in Fig.~\ref{fig:L_SEP_JEP}, we show that the second-order region under JEP is slightly larger than SEP, which implies that JEP is a better criterion.

Finally, we generalize our results under SEP to the case with quasi-static fading  where both users have full channel state information (CSI). We remark that such a assumption is common in theoretical analysis \cite{zhou2019JSAC}~\cite{Hoydis2015fading}.We show that second-order rates are equal to zero, which is consistent with existing results~\cite{yang2014quasi,zhou2019JSAC}. Our results imply that the asymptotic notion of outage capacity and outage probability are valid at finite blocklength under the mismatched coding scheme.

\subsection{Organization of the Rest of the Paper}
The rest of paper is organized as follows. In Section~\ref{sec:problem_fomulation}, we set up the notation and formulate the ABC problem. In Section~\ref{sec:main_rsults}, we present achievable second-order asymptotics under both SEP and JEP. In Section~\ref{sec:fading}, we present corresponding results for the quasi-static fading setting. The proofs of our results are presented in Sections~\ref{sec:proof_sketch_SEP}-\ref{sec:proof_fading}. Finally, in Section~\ref{sec:conclusion}, we conclude the paper and discuss future research directions.

\section{Problem Formulation and Definitions}
\label{sec:problem_fomulation}
\subsection*{Notation}
Random variables are in capital case (e.g., $X$) and their realizations are in lower case (e.g., $x$). Random vectors of length $n$ and their particular realizations are denoted by $X^n:= (X_1, \ldots, X_n)$ and $x^n=(x_1,\ldots,x_n)$, respectively. We use calligraphic font (e.g., $\mathcal{X}$) to denote all sets.
We use $\bbR$, $\bbR_+$, $\bbN$ to denote the set of real numbers, positive real numbers and integers respectively. For any two integers $(a,b)\in\bbN^2$, we use $[a:b]$ to denote the set of integers between $a$ and $b$, and we use $[a]$ to denote $[1:a]$. For any $(a,b)\in\bbR^2$ and $n\in\bbN$, we use $(a,b)$ to denote the set of all real numbers between $a$ and $b$, and we use $(a,b)^n$ to denote the $n$-fold Cartesian product. We use logarithm with base $e$. The Gaussian complementary cumulative distribution function (cdf) and its inverse are denoted by $Q(\cdot)$ and $Q^{-1}(\cdot)$, respectively. We use $\| x^n \| = \sqrt {\sum\nolimits_i {x_i^2} } $ to denote the $\ell_2$ norm of a vector $x^n \in \mathbb{R}^n$. We use $1\{\cdot\}$ to denote the indicator function, i.e., $1\{\rmA\}=1$ if A is true and $1\{\rmA\}=0$ otherwise. For any $(x,\mu,\sigma)\in\bbR^3$, we use $\calN(x;\mu,\sigma^2)$ to denote a Gaussian distribution with mean $\mu$ and variance $\sigma^2$ with $x$ being a particular realization.

\subsection{Problem Formulation}
Fix any blocklength $n\in\bbN$. Consider the following two-user ABC with additive non-Gaussian noise:
\begin{align}
Y_1^n&=X^n+Z_1^n,\\
Y_2^n&=X^n+Z_2^n,
\end{align}
where $X^n$ is the channel input of the transmitter, $Y_1^n$ and $Y_2^n$ are the received signals of users 1 and 2, respectively. For each $i\in[2]$, $Z_{i}^n$ is the additive noise, which is generated i.i.d. from an arbitrary distribution $P_Z$ satisfying the following moment constraints
\begin{align}
\begin{array}{lll}
\mathbb{E}[Z_1^2]=\beta,\;&\zeta_1:=\mathbb{E}[Z_1^4]<\infty,\;&\mathbb{E}[Z_1^6]<\infty,\\
\mathbb{E}[Z_2^2]=1,\;&\zeta_2:=\mathbb{E}[Z_2^4]<\infty,\;&\mathbb{E}[Z_2^6]<\infty.
\end{array}\label{def:noise}
\end{align}
The constraints in \eqref{def:noise} are consistent with~\cite{scarlett2017mismatch,zhou2019jscc}, which are necessary to derive our theoretical benchmarks. Several common distributions of noise satisfy the constraints in \eqref{def:noise} including Gaussian, Laplace and uniform. Without loss of generality, we assume that the noise power of user 1 is $\beta\in(0,1)$, which is strictly smaller than the unity noise power of user 2. Thus, user 1 and 2 are the strong user and weak user, respectively.

For the weak user, nearest neighbor decoding and treating interference as noise are used. For the strong user, we consider both SIC and JNN decoding schemes. We first present the definition of a code when SIC is used by the strong user.
\begin{definition}\label{def:code_SIC}
An $(n,M_1,M_2)$-SIC code for a two user ABC consists of
\begin{itemize}
\item a set of $M_1$ codewords $\{V^n(i)\}_{i\in[M_1]}$ known by encoder and decoder of user 1
\item a set of $M_2$ codewords $\{U^n(i)\}_{i\in[M_2]}$ known by encoder and decoders of both users
\item encoder $f$ such for each $(w_1,w_2)\in[M_1]\times[M_2]$:
\begin{align}
f(w_1,w_2):&=X^n(w_1,w_2)\\
&=V^n(w_1)+U^n(w_2),
\end{align}
where $X^n(w_1,w_2)$ is the out put of the encoder $f$ and for each $i\in[2]$, $w_i$ is the message that the encoder sends to user $i$,
\item a decoder $\phi_2$ for user 2 that applies nearest neighbor decoding by treating the signal of user 1 as noise:
\begin{align}
\hatW_2:=\argmin_{w_2\in[M_2]}\|Y_2^n-U^n(w_2)\|^2, \label{def:phi2_NN}
\end{align}
\item a nearest neighbor SIC decoder $\phi_1$ for user 1 that operates as follows. Firstly, it decodes the message for user 2 while treating the signal of user 1 as noise  %$\phi_1(Y_1^n):=\hatW_1$
\begin{align}
\barW_2:=\argmin_{w_2\in[M_2]}\|Y_1^n-U^n(w_2)\|^2. \label{def:phi1_TIN}
\end{align}
Subsequently, it decodes the message for itself using another NN decoder with SIC by subtracting  $U^n(\barW_2)$,
\begin{align}
\hatW_1:=\argmin_{w_1\in[M_1]}\|Y_1^n-U^n(\barW_2)-V^n(w_1)\|^2. \label{def:phi1_NN}
\end{align}
\end{itemize}
\end{definition}

An $(n,M_1,M_2)$-JNN code for a two user ABC is similar to an $(n,M_1,M_2)$-SIC code defined above, with the only difference being that we replace the SIC decoder in \eqref{def:phi1_TIN} and \eqref{def:phi1_NN} with the following JNN decoder
\begin{align}
(\hat{W}_1,\bar{W}_2):=\argmin_{(w_1,w_2)\in[M_1]\times[M_2]}\|Y_1^n-U^n(w_2)-V^n(w_1)\|^2.
\end{align}

The transmitter uses a spherical codebook for each user. Specifically, given any $x^n=(x_1,\ldots,x^n)\in\bbR^n$ and $P\in\bbR_+$, we define the following uniform distribution over the surface of a sphere with center $\mathbf{0}^n$ and radius $\sqrt{nP}$:
\begin{align}
f_{\mathrm{sp}}(x^n|P)=\frac{1\{\|x^n\|^2-nP\}}{S_n(\sqrt{nP})},
\end{align}
where $1\{\cdot\}$ is the indicator function, $S_n(r)=n\pi^{n/2}r^{n-1}/\Gamma( {\frac{n+2}{2}})$ is the surface area of an $n$-dimensional sphere with radius $r$, and $\Gamma(\cdot)$ is the Gamma function. We choose spherical codebooks rather than i.i.d. Gaussian codebooks since the latter one was shown to have worse finite $n$ performance in the point-to-point non-Gaussian noise channel~\cite[Section II. C]{tan2014dispersions}. Let $\alpha\in(0,1)$ and $\bar{\alpha}=1-\alpha$. The codewords $\{V^n(i)\}_{i\in[M_1]}$ for user 1 and $\{U^n(i)\}_{i\in[M_2]}$ for user 2 are generated independently from the following distributions: for any $(v^n.u^n)\in(\bbR^n)^2$,
\begin{align}
P_{V^n}(v^n)&:= f_{\mathrm{sp}}(v^n|\alpha P),\label{def:PV^n}\\
P_{U^n}(u^n)&:= f_{\mathrm{sp}}(u^n|\bar{\alpha} P). \label{def:PU^n}
\end{align}
With the above choice of the power of the codewords, the sum power of the transmit signal equals to $nP$ for an $(n,M_1,M_2)$-SIC or an $(n,M_1,M_2)$-JNN code.

To evaluate the performance of the coding scheme, we consider ensemble separate error probabilities (SEP):
\begin{align}
\rmP_{\rme,1}^n&:=\Pr\{\hatW_1\neq W_1\;\mathrm{or}\;\barW_2\neq W_2\},\label{def:Pe1}\\
\rmP_{\rme,2}^n&:=\Pr\{\hatW_2\neq W_2\},\label{def:Pe2}
\end{align}
and ensemble joint error probability (JEP):
\begin{align}
\rmP_{\rme,\rmJ}^n:=\Pr\{\hatW_1\neq W_1\;\mathrm{or}\;\barW_2\neq W_2 \;\mathrm{or}\;\hatW_2\neq W_2\}. \label{def:PeJ}
\end{align}
The above probability terms average not only over the distribution of the messages $(W_1,W_2)$, but also over the distribution of the codewords $\{V^n(i)\}_{i\in[M_1]}$ for user 1 and $\{U^n(i)\}_{i\in[M_2]}$ and channel noises $(Z_1^n,Z_2^n)$. These definitions are consistent with existing studies on mismatched communication~\cite{lapidoth1996,lapidoth1997,scarlett2017mismatch,zhou2018refined}.

When the private message for user 2 is absent, the error probabilities of two users are asymmetric. In this case, \eqref{def:Pe1} and \eqref{def:PeJ} are replaced by
\begin{align}
\bar\rmP_{\rme,1}^n&:=\Pr\{\hatW_1\neq W_1\}, \label{def:barPe1}\\
\bar\rmP_{\rme,J}^n&:=\Pr\{\hatW_1\neq W_1,\hatW_2\neq W_2\}, \label{def:barPeJ}
\end{align}
respectively, and ABC is now a SBC. We remark that our achievability results also hold for SBC (see discussions in the second to last remark of Theorem~\ref{theo:SEP}).

\subsection{Achievable Rate Region}
The achievable rate region of ABC collects the rate pairs of message pairs such that our codes in Definition \ref{def:code_SIC} ensures vanishing ensemble error probabilities. Let $\dagger:=\{\mathrm{SIC},\mathrm{JNN}\}$. The achievable rate region under SEP is defined as follows.
\begin{definition}\label{def:first_region}
A rate pair $(R_1,R_2)\in\bbR_+^2$ is said to be $(P,\dagger)$-achievable for mismatched ABC under SEP if there exists a sequence of $(n,M_1,M_2)$-$\dagger$ codes such that
\begin{align}
\limsup_{n\to\infty}\frac{1}{n}\log M_i&\geq R_i,
\end{align}
and
\begin{align}
&\max_{i\in[2]}\lim_{n\to\infty}\rmP_{e,i}^n=0.
\end{align}
The convex closure of the set of all $(P,\dagger)$-achievable rate pairs under SEP is denoted by $\calR^{\rm{SEP}}_{\dagger}(P)$.
\end{definition}

We next define the achievable rate region under JEP.
\begin{definition}
A rate pair $(R_1,R_2)\in\bbR_+^2$ is said to be $(P,\dagger)$-achievable for mismatched ABC under JEP if there exists a sequence of $(n,M_1,M_2)$-$\dagger$ codes such that
\begin{align}
\limsup_{n\to\infty}\frac{1}{n}\log M_i&\geq R_i,
\end{align}
and
\begin{align}
&\lim_{n\to\infty}\rmP_{\rme,\rmJ}^n=0.
\end{align}
The convex closure of the set of all $(P,\dagger)$-achievable rate pairs under JEP is denoted by $\calR^{\rm{JEP}}_{\dagger}(P)$.
\end{definition}

Now, we present the achievable rate region under SEP. Let $\rmC(P):=\frac{1}{2}\log(1+P)$ denote the Gaussian channel capacity. As a corollary of our results in Theorem~\ref{theo:SEP} by letting $n\to\infty$, the following inner bound to $\calR_{\dagger}^{\rm{SEP}}(P)$ holds:
\begin{align}
\nn\calR_{\mathrm{inner}}(P):=&\bigcup_{\alpha\in(0,1)}\Big\{(R_1,R_2):\\*
&\quad R_1\leq\rmC\Big(\frac{\alpha P}{\beta}\Big),
R_2\leq\rmC\Big(\frac{\bar\alpha P}{\alpha P+1}\Big)\Big\}.
\end{align}
Note that the above rate-region also holds for JEP since $\rmP_{\rme,\rmJ}^n\leq\rmP_{\rme,1}^n+\rmP_{\rme,2}^n$. In Fig. \ref{fig:L1L2_plot}, we numerically illustrate the first-order achievable region $\calR_{\mathrm{inner}}(P)$. Note that $\calR_{\mathrm{inner}}(P)$ is the optimal achievable region for the Gaussian BC~\cite{bergmans1974broadcast}.

\begin{figure}
\centering
\includegraphics[width=.8\columnwidth]{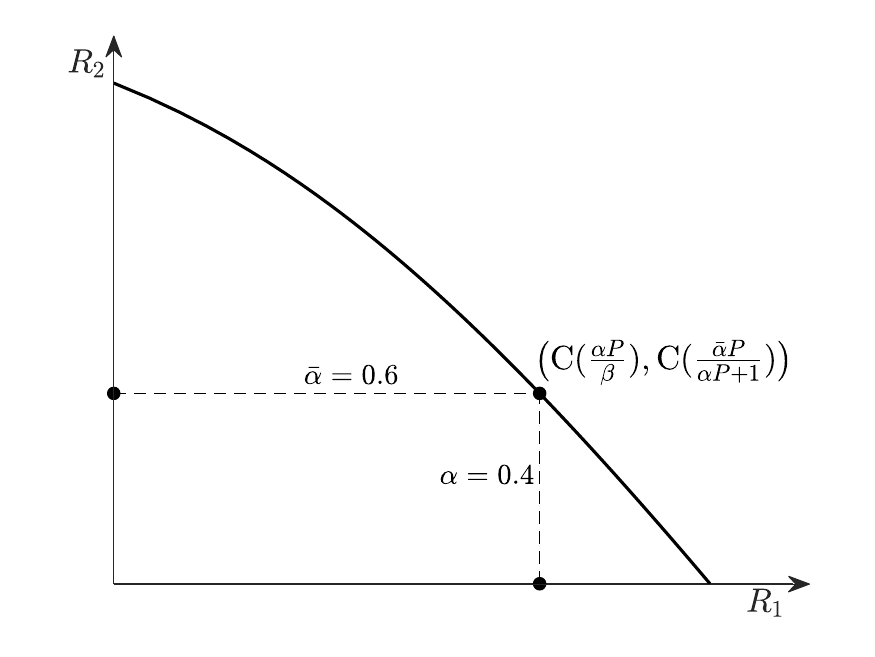}
\caption{Illustration of the achievable rate region. Given a particular pair $(\alpha,\bar\alpha)$, the achievable region is characterized by a rectangular. The whole rate region is obtained by taking the union over all $\alpha\in(0,1)$.}
\label{fig:L1L2_plot}
\end{figure}

We now present the definition of the second-order asymptotics under the SEP criterion.
\begin{definition}
Given any $(\varepsilon_1,\varepsilon_2)\in(0,1)^2$, $P\in\bbR_+$ and $\alpha\in(0,1)$, a pair $(L_1,L_2)\in\bbR_+^2$ is said to be second-order $(\varepsilon_1,\varepsilon_2,\alpha,P,\dagger)$-achievable under SEP if there exists a sequence of $(n,M_1,M_2)$-$\dagger$ codes such that
\begin{align}
\limsup_{n\to\infty}\frac{1}{\sqrt{n}}\left(n\rmC\Big(\frac{\alpha P}{\beta}\Big)-\log M_1\right)&\leq L_1,\\
\limsup_{n\to\infty}\frac{1}{\sqrt{n}}\left(n\rmC\Big(\frac{\bar\alpha P}{\alpha P+1}\Big)-\log M_2\right)&\leq L_2,
\end{align}
and
\begin{align}
\lim_{n\to\infty}\rmP_{\rme,i}^n\leq\varepsilon_i~\forall~i\in[2].
\end{align}
The convex closure of the set of all second-order $(\varepsilon_1,\varepsilon_2,\alpha,P,\dagger)$-achievable rate pairs under SEP is denoted by $\calL_{\rm{SEP}}(\varepsilon_1,\varepsilon_2,\alpha,P,\dagger)$.
\end{definition}

Finally, we present the definition of the second-order asymptotics under the JEP criterion.
\begin{definition}\label{def:JEP_achL1L2}
For $i\in[2]$, given any $\varepsilon\in(0,1)$ and $\alpha\in(0,1)$, a pair $(L_1,L_2)\in\bbR_+^2$ is said to be second-order $(\varepsilon,\alpha,P,\dagger)$-achievable under JEP if there exists a sequence of $(n,M_1,M_2,\alpha,P)$-$\dagger$ codes such that
\begin{align}
\limsup_{n\to\infty}\frac{1}{\sqrt{n}}\left(n\rmC\Big(\frac{\alpha P}{\beta}\Big)-\log M_1\right)&\leq L_1,\\
\limsup_{n\to\infty}\frac{1}{\sqrt{n}}\left(n\rmC\Big(\frac{\bar\alpha P}{\alpha P+1}\Big)-\log M_2\right)&\leq L_2,
\end{align}
and
\begin{align}
\lim_{n\to\infty}\rmP_{\rme,\rmJ}^n\leq\varepsilon,
\end{align}
The convex closure of the set of all second-order $(\varepsilon,\alpha,P,\dagger)$-achievable rate pairs under JEP is denoted by $\calL_{\rm{JEP}}(\varepsilon,\alpha,P,\dagger)$.
\end{definition}

\section{Main Results}\label{sec:main_rsults}
In this section, we present our main results under SEP. To do so, we need the following definitions. First, define the following two dispersion functions:
\begin{align}
\rmV_1(P,\beta,\zeta)&:=\frac{P^2(\zeta-\beta^2)+4P\beta^3}{4\beta^2(P+\beta)^2}, \label{def:V1()}\\
\rmV_2(P,\barP,\beta,\zeta)&:=\frac{P^2(\zeta-\beta^2+4\barP)+4P(\barP +\beta)^3}{4(\barP+\beta)^2(P+\barP+\beta)^2}. \label{def:V2()}
\end{align}
Recall $(\zeta_1,\zeta_2)$ in constraints of the noise distribution in \eqref{def:noise}. We will use $\rmV_1(\alpha P,\beta,\zeta_1)$ and $\rmV_2(\bar\alpha P,\alpha P,1,\zeta_2)$, which are abbreviated as $\rmV_1(\alpha P)$ and $\rmV_2(\bar\alpha P, \alpha P)$, respectively for ease of notation. We also need the following dispersion matrix:
\begin{align}
\bV(\alpha,P)&:=
\begin{bmatrix}
\rmV_1(\alpha P)&0\\
0&\rmV_2(\bar\alpha P,\alpha P)
\end{bmatrix}. \label{def:bfV}
\end{align}
Finally, given any $d\in\bbN$ and any $\varepsilon\in(0,1)$, we need the following complementary Gaussian cumulative distribution function (ccdf) of a  Gaussian random vector with zero mean vector and the covariance matrix $\bV$:
\begin{align}
\rmQ_{\mathrm{inv}}(\mathbf{V},\varepsilon):=\{\mathbf{s}_d\in\mathbb{R}^d: \Pr\{\mathbf{S}_d\leq\mathbf{s}_d\}\geq1-\varepsilon\}, \label{def:Qinv}
\end{align}
where $\mathbf{S}_d\sim\calN_d(\mathbf{0},\mathbf{V})$ denotes the $d$-dimension Gaussian random vector.

An achievability result for the second-order region of the mismatched ABC under SEP is characterized by following theorem.
\begin{theorem}\label{theo:SEP}
Given any $(\varepsilon_1,\varepsilon_2)\in(0,1)^2$, $P\in\bbR_+$ and $\alpha\in(0,1)$, for any $\dagger\in(\mathrm{SIC},\mathrm{JNN})$,
\begin{align}
\nn\Big\{(L_1,L_2):L_1&\geq \sqrt{\rmV_1(\alpha P)}\rmQ^{-1}(\varepsilon_1),\\
\nn L_2&\geq \sqrt{\rmV_2(\bar\alpha P,\alpha P)}\rmQ^{-1}(\varepsilon_2)\Big\}\\
&\qquad\subseteq\calL_{\rm{SEP}}(\varepsilon_1,\varepsilon_2,\alpha,P,\dagger).
\end{align}
\end{theorem}
The proof of Theorem~\ref{theo:SEP} is provided in Section \ref{sec:proof_sketch_SEP}. We make the following remarks.

The proof of Theorem~\ref{theo:SEP} under SIC extends the proof of the mismatched point-to-point channel~\cite[Theorem 1]{scarlett2017mismatch} and interference channel~\cite[Theorem 2]{scarlett2017mismatch} to the ABC setting. Specifically, user 2 tries to recover the message $W_2$ while treats $V^n$ as noise, which is equivalent  to the receiver of the mismatched interference channel~(cf. \cite[Section III]{scarlett2017mismatch}) with one interference signal. Note that user 1 adopts the SIC scheme by firstly operating as the same manner to user 2, then operating as a point-to-point decoder after subtracting the signal for user 2. However, the proof of the case under JNN does not directly follow from that of SIC, since the decoder $\psi_2$ aims to recover both messages $W_1$ and $W_2$ simultaneously. Inspired by the analysis of multiple access channel with non-Gaussian noise~\cite[Theorem 3]{wang2023achievable}, we manage to solve the problem by extending the analysis of \cite[Theorem 3]{wang2023achievable} to the ABC setting.

Theorem 1 shows that the second-order rate regions under SIC and JNN are identical, which was somewhat surprising since the joint decoding scheme usually has better performance than a separate scheme. Note that the decoding scheme for user 2 is exactly the same, which ensures the same second-order rate $L_2$ for user 2. In contrast, for user 1 under SIC, the first step where the codewords $V^n$ are treated as interference is asymptotically error free since the channel of user 1 is strictly better than user 1. Thus, the error event mainly results from the process recovering $W_2$ which leads to a bound on $L_1$. For user 1 under JNN, the decoder $\psi_2$ searches the indices $(\hatW_1,\barW_2)$ in a two-dimensional surface in order to decode the message for both users simultaneously. Such a scheme is supposed to generate a constraint of $\log M_1$ and $\log M_2$.  However, user 1 has a better channel than user 2 while the size of the code designed for user 2, denoted by $\log M_2$, is determined by the ``worse'' channel. Thus, the constraint of $\log M_2$ is satisfied almost surely. This way, the bound on $\log M_1$ is determined by a one-dimensional search, which is identical to the case under SIC. Such a result implies that one should always apply the SIC decoding scheme since the computational complexity of SIC decoding is $O(n)$, which is much smaller than the complexity of JNN decoding is $O(n^2)$.

Our result can be generalized to the SBC, where each user decodes one message only (cf. \eqref{def:Pe2} and \eqref{def:barPe1}). Our achievability proof holds since the error probability of user 1 for ABC in \eqref{def:Pe1} is a natural upper bound for the error probability of user 1 for a SBC in  \eqref{def:barPe1}.

When specialized to a ABC with additive Gaussian noise, it follows that $\zeta_1=3\beta^2$ and $\zeta_2=3$. Such a result generalizes the second-order asymptotics for a discrete memoryless ABC~\cite[Section IV]{tan2014dispersions} to the Gaussian setting.  Furthermore, by taking $n\to\infty$, our results recovers the first-order asymptotic region~\cite{bergmans1973broadcast,bergmans1974broadcast}.

We next present our results under JEP.
\begin{theorem}\label{theo:JEP}
Given any $\varepsilon\in(0,1)$, $P\in\bbR_+$ and $\alpha\in(0,1)$, for any $\dagger\in(\mathrm{SIC},\mathrm{JNN})$,
\begin{align}
&\nn\Big\{(L_1,L_2):(L_1,L_2)\in\rmQ_{\mathrm{inv}}(\mathbf{V}(\alpha,P),\varepsilon)\Big\}\\
&\qquad\qquad\qquad\qquad\qquad\subseteq\calL_{\mathrm{JEP}}(\varepsilon,\alpha,P,\dagger). \label{JEP:theorem}
\end{align}
\end{theorem}
The proof of Theorem~\ref{theo:JEP} is provided in Section \ref{sec:proof_JEP}. A few remarks are in order.

The proof of Theorem~\ref{theo:JEP} are inspired by the mismatched multiple access channel in~\cite[Theorem 3]{wang2023achievable}. Specifically, we firstly generalize the random coding union (RCU) bound~\cite[Theorem 16]{polyanskiy2010finite} to the ABC setting under both SIC and JNN decoding schemes. Subsequently, we upper bound the error probability and establish the relationship between the error probability and the size of the codebooks. Note that the mismatched ABC is a degraded broadcast channel~\cite[Chapter 5.5]{el2011network} since the signal received by the weak user can be considered as the signal received by the strong user after passing through channel with power $1-\beta$. Due to this nature, the constraints on the sum rate can be reduced if the constraints on $M_1$ and $M_2$ are satisfied simultaneously (cf. \eqref{JEP:JNN_factM2}-\eqref{JEP:JNN_donesimplify}).

Next we compare the results under SEP and JEP criteria. Specifically, if $L_2\to\infty$, the left hand side of \eqref{JEP:theorem} degrades to
\begin{align}
\big\{(L_1,L_2):L_1\geq \sqrt{\rmV_1(\alpha P)}\rmQ^{-1}(\varepsilon)\big\}.\label{JEP:theo_degradedL2}
\end{align}
Correspondingly, if $L_1\to\infty$, the left hand side of \eqref{JEP:theorem} degrades to
\begin{align}
\big\{(L_1,L_2):L_2\geq \sqrt{\rmV_2(\bar\alpha P,\alpha P)}\rmQ^{-1}(\varepsilon)\big\}. \label{JEP:theo_degradedL1}
\end{align}
Letting $\varepsilon=\varepsilon_1+\varepsilon_2$, the above results match the achievable under SEP. In Fig.~\ref{fig:L_SEP_JEP}, we numerically illustrate this point.

We next discuss non-extreme case where neither $L_1$ nor $L_2$ tends to infinity. Note that the dispersion matrix $\bV(\alpha,P)$ in \eqref{def:bfV} is diagonal, which seems to imply that the second-order rates of user 1 and 2 are independent. However, taking into account the fact that $L_1$ and $L_2$ are also determined by the joint error probability $\varepsilon$, there exists a tradeoff between $L_1$ and $L_2$. Specifically, we compare the second-order rate regions under SEP and JEP by presenting a numerical result in Fig.~\ref{fig:L_SEP_JEP}. Under SEP, we assume that the error probabilities of user 1 and 2 are constrained by a sum (joint) error probability, i.e. $\varepsilon_1+\varepsilon_2=0.3$ while correspondingly, we assume the joint error probability satisfies that $\varepsilon=0.3$ under JEP. It is shown that the second-order rate regions under SEP and JEP are matched when $L_1\to\infty$ or $L_2\to\infty$, which corroborates \eqref{JEP:theo_degradedL2} and \eqref{JEP:theo_degradedL1}. However, as shown in Fig.~\ref{fig:L_SEP_JEP}, the rate regions under SEP and JEP do not match when $L_1$ and $L_2$ are finite, i.e., the result under SEP is bounded by that under JEP. This implies that considering JEP criterion brings some advantages if one aims to measure the performance of the mismatched ABC by the sum error probability of two users.

Similar to Theorem~\ref{theo:SEP}, the second-order rate regions under SIC and JNN are same and the similar remark (cf. third paragraph after Theorem~\ref{theo:SEP}) also holds for JEP.

\begin{figure}[tb]
\centering
\includegraphics[width=.9\columnwidth]{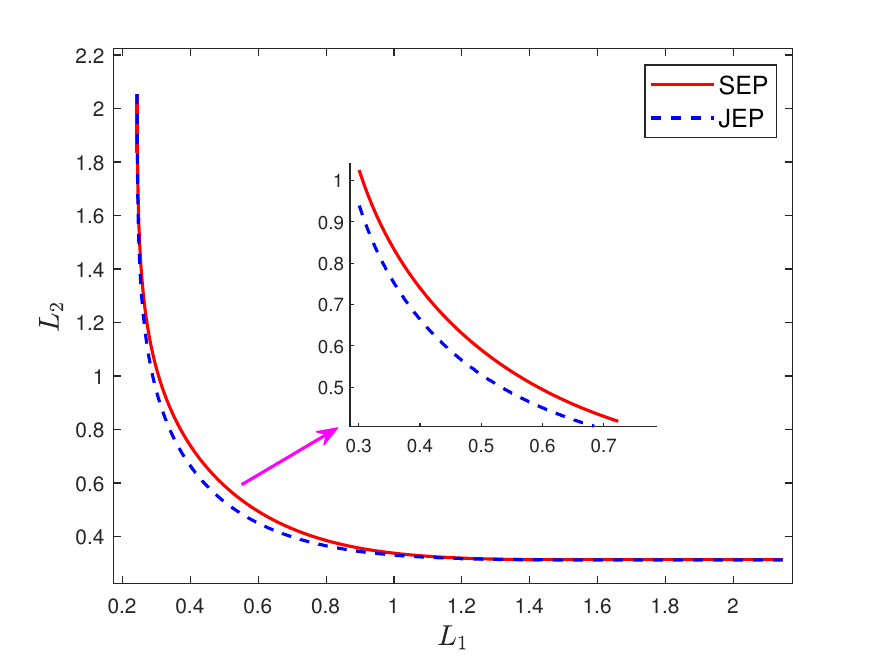}
\caption{Illustration of second-order rate SEP and JEP when $P=5$, $\alpha=0.3$, $\bar\alpha=0.7$, $\beta=0.6$, $\zeta_1=0.3888$ and $\zeta_2=3$. }
\label{fig:L_SEP_JEP}
\end{figure}
%\vspace{-.5cm}

\section{Mismatched ABC with quasi-static fading}\label{sec:fading}
In this section, we generalize our results to the case with quasi-static fading. Specifically, for each $i\in[2]$, the fading coefficient $H_i\in\calH$ are generated according to a distribution $P_{H_i}$ and remains unchanged in $n$ channel uses. Thus, the channel outputs for two users satisfy
\begin{align}
Y_1^n&=H_1X^n+Z_1^n,\\
Y_2^n&=H_2X^n+Z_2^n,
\end{align}
where the additive noises $Z_1^n$ and $Z_2^n$ satisfy the same moment constraints in \eqref{def:noise}. We assume that only the receiver $i$ has access to its own CSI, i.e., the fading coefficient $H_i$. Similar to the case without channel fading, the weak user (user 2) treats interference as noise while the strong user (user 1) employs two decoding schemes: SIC and JNN. We next define codes.

The codewords allocation and encoder of an $(n,M_1,M_2)$-SIC code for a two user ABC with quasi-static fading remains the same as that of an $(n,M_1,M_2)$-SIC code for a two user ABC (cf. Definition~\ref{def:code_SIC}). With the CSI of user 2, the decoder $\bar\phi_2$ treats the signal of user 1 as noise using the following NN decoder
\begin{align}
\hatW_2:=\argmin_{w_2\in[M_2]}\|Y_2^n-H_2U^n(w_2)\|^2. \label{def:fading_phi2_NN}
\end{align}
With the CSI of user 1, the decoder $\bar\phi_1$ for user 1 operates with two steps. Firstly, it decodes the message for user 2 while treating the signal of user 1 as noise using the following NN decoder
\begin{align}
\barW_2:=\argmin_{w_2\in[M_2]}\|Y_1^n-H_1U^n(w_2)\|^2\label{def:fading_phi1_TIN}.
\end{align}
Subsequently, user 1 decodes the message for itself by subtracting $H_1U^n(\barW_2)$ with another NN decoder
\begin{align}
\hatW_1:=\argmin_{w_1\in[M_1]}\|Y_1^n-H_1U^n(\barW_2)-H_1V^n(w_1)\|^2. \label{def:fading_phi1_NN}
\end{align}

An $(n,M_1,M_2)$-JNN code for a two user ABC with quasi-static fading is similar to an $(n,M_1,M_2)$-SIC code defined above, with only exception of replacing the SIC decoder in \eqref{def:fading_phi1_TIN} and \eqref{def:fading_phi1_NN} with the following JNN decoder
\begin{align}
\nn&(\hat{W}_1,\bar{W}_2)\\*
&\quad:=\argmin_{(w_1,w_2)\in[M_1]\times[M_2]}\|Y_1^n-H_1U^n(w_2)-H_1V^n(w_1)\|^2.
\end{align}

We use the SEP criterion (cf. \eqref{def:Pe1}-\eqref{def:Pe2}) to evaluate the performance of a code. The results for JEP can be obtained analogously. Furthermore, the following moment assumptions on the fading distribution are required:
\begin{align}
\mathbb{E}_{H_1}[H_1^9]<\infty,\quad \mathbb{E}_{H_2}[H_2^{12}]<\infty. \label{def:fading_SEP_assump_H}
\end{align}
We remark that such assumptions are necessary for the use of Berry-Esseen theorem and are satisfied by commonly assumed models for channel fading, e.g., Rayleigh fading and Rice fading. Similar assumptions can be found in existing studies~\cite[Theorem 3]{yang2014quasi}~\cite[Theorem 5]{zhou2019JSAC}.

Given fading parameters $H_1$ and $H_2$, and the power of codewords $\alpha P$ and $\bar\alpha P$, let $\rmC_1(H_1,\alpha P):=\rmC\left(\frac{H_1^2\alpha P}{\beta}\right)$ and $\rmC_2(H_2,\alpha P,\bar\alpha P):=\rmC\left(\frac{H_2^2\bar\alpha P}{H_2^2\alpha P+1}\right)$. Furthermore, let $\nu_n:=\frac{\log n}{\sqrt{n}}$. Our result for the $(n,M_1,M_2)$-SIC code can be stated as follows.

\begin{theorem}\label{theo:fading_SEP}
Under the assumptions in \eqref{def:fading_SEP_assump_H}, an $(n,M_1,M_2)$-SIC code satisfies
\begin{align}
\rmP_{\rme,1}^n&\leq\mathbb{E}_{H_1}\left[\rmQ\left(\frac{n\rmC_1(H_1,\alpha P)-\log M_1}{\sqrt{n \mathrm{V}_1(H_2^2\alpha P)}}\right)\right]+O(\nu_n),\label{theo:fading_Pe1}\\
\rmP_{\rme,2}^n&\leq\mathbb{E}_{H_2}\left[\rmQ\left(\frac{n\rmC_2(H_2,\alpha P,\bar\alpha P)-\log M_2}{\sqrt{n \mathrm{V}_2(H_2^2\bar\alpha P,H_2^2\alpha P)}}\right)\right]+O(\nu_n). \label{theo:fading_Pe2}
\end{align}
\end{theorem}
The proof sketch of Theorem~\ref{theo:fading_SEP} is provided in Section~\ref{sec:proof_fading}.

Theorem \ref{theo:fading_SEP} implies that the error probabilities of both users are dominated by the expectation terms involving the ccdf $\rmQ$ function, with expectation over the fading coefficients. Using~\cite[Lemma 17]{yang2014quasi}, the dominated terms equal to the outage probabilities at finite block-length, i.e.,
\begin{align}
\nn&\mathbb{E}_{H_1}\left[\rmQ\left(\frac{n\rmC_1(H_1,\alpha P)-\log M_1}{\sqrt{n \mathrm{V}_1(H_2^2\alpha P)}}\right)\right]\\
&\qquad=\rmP_{\mathrm{out},1}^n(M_1)+O(\nu_n), \label{SEP:fading_usingLemma17_1}\\
\nn&\mathbb{E}_{H_2}\left[\rmQ\left(\frac{n\rmC_2(H_2,\alpha P,\bar\alpha P)-\log M_2}{\sqrt{n \mathrm{V}_2(H_2^2\bar\alpha P,H_2^2\alpha P)}}\right)\right]\\
&\qquad=\rmP_{\mathrm{out},2}^n(M_2)+O(\nu_n), \label{SEP:fading_usingLemma17_2}
\end{align}
where the outage probabilities are defined as
\begin{align}
\rmP_{\mathrm{out},1}^n(M_1)&:=\Pr\{n\rmC_1(H_1,\alpha P)\leq\log M_1\},\\
\rmP_{\mathrm{out},2}^n(M_2)&:=\Pr\{n\rmC_2(H_2,\alpha P,\bar\alpha P)\leq\log M_2\}.
\end{align}

We next present achievable rate regions using outage probabilities. Given any $(\varepsilon_1,\varepsilon_2)\in(0,1)^2$, define the outage capacity region for mismatched ABC under SEP as
\begin{align}
\nn\calR_{\mathrm{out}}^{\mathrm{SEP}}(\varepsilon_1,\varepsilon_2):=&\{(R_1,R_2)\in\bbR_+: \rmP_{\mathrm{out},1}^n(\exp\{nR_1\})\leq \varepsilon_1,\;\\*
&\qquad\rmP_{\mathrm{out},2}^n(\exp\{nR_2\})\leq \varepsilon_2\}.
\end{align}
\begin{corollary}\label{theo:coro_fading_SEP}
Under the conditions of Theorem~\ref{theo:fading_SEP}, given any $(\varepsilon_1,\varepsilon_2)\in(0,1)^2$ and any rate pair $(R_1,R_2)\in\calR_{\mathrm{out}}^{\mathrm{SEP}}(\varepsilon_1,\varepsilon_2)$, there exists an $(n,M_1,M_2)$-code such that for each $i\in[2]$,
\begin{align}
\rmP_{\mathrm{e},i}^n&\leq \varepsilon_i,\\
\log M_i&=nR_i+O(\log n).
\end{align}
\end{corollary}

The proof of Corollary~\ref{theo:coro_fading_SEP} is similar to~\cite[Corollary 6]{zhou2019JSAC}, which follows from the definition of $\calR_{\mathrm{out}}^{\mathrm{SEP}}(\varepsilon_1,\varepsilon_2)$, \eqref{SEP:fading_usingLemma17_1}, \eqref{SEP:fading_usingLemma17_2} and Theorem~\ref{theo:fading_SEP}.

Corollary~\ref{theo:coro_fading_SEP} implies that the second-order rates with the order $1/\sqrt{n}$ equal to 0, which are consistent with existing results in~\cite{yang2014quasi,zhou2019JSAC}. Furthermore, Corollary~\ref{theo:coro_fading_SEP} indicates that the outage capacity region, a widely used asymptotic notion in wireless communication systems under quasi-static fading, is an accurate performance criterion at finite blocklength for our mismatched coding schemes.

Furthermore, we remark that the corresponding result under JEP is identical to Corollary~\ref{theo:coro_fading_SEP} with the only exception that the separate error probabilities $(\rmP_{\rme,i}^n,\varepsilon_i)$ are replaced by the joint error probability $(\rmP_{\rme,\rmJ},\varepsilon)$.

\section{Proof of Theorem~\ref{theo:SEP}}\label{sec:proof_sketch_SEP}
\subsection{Preliminaries}
Recall that $\rmC(P):=\frac{1}{2}\log(1+P)$ and $\calN(x;\mu,\sigma^2)$ is a Gaussian distribution with mean $\mu$ and variance $\sigma^2$ with $x$ being a particular realization. Our analysis makes extensive use of ``mismatched'' information density~\cite[Section IV]{scarlett2017mismatch}. Specifically, given any $(u,y_2)\in\bbR^2$, let
\begin{align}
\tilde\imath(u;y_2):&=\log\frac{\calN(y_2;u,1+\alpha P)}{\calN(y;0,1+P)}\\
&=\rmC\left(\frac{\bar\alpha P}{\alpha P+1}\right)+\frac{y_2^2}{2(P+1)}-\frac{(y_2-u)^2}{2(\alpha P+1)}, \label{def:density_uy2}
\end{align}
which is the information density of the Gaussian channel $\calN(y;u,\alpha P+1)$ with input $u$, noise power $P+1$ and output $y$. For any $n\in\bbN$, we define the following $n$-letter version of the mismatched information density
\begin{align}
\tilde\imath^n(u^n;y_2^n):=\sum_{i=1}^n\tilde\imath(u_i;y_{2,i}).
\end{align}

Similarly, for any $(y_1,v,u)\in\bbR^3$ and any $n\in\bbN$, we define the conditional information densities
\begin{align}
&\tilde\imath(v;y_1|u):=\rmC\left(\frac{\alpha P}{\beta}\right)+\frac{(y_1-u)^2}{2(\alpha P+\beta)}-\frac{(y_1-u-v)^2}{2\beta},\\
&\tilde\imath^n(v^n;y_1^n|u^n):=\sum_{i=1}^n\tilde\imath(v_i;y_{1,i}|u_i),
\end{align}
and the joint information densities
\begin{align}
&\tilde\imath(v,u;y_1):=\rmC\left(\frac{P}{\beta}\right)+\frac{y_1^2}{2(P+\beta)}-\frac{(y_1-u-v)^2}{2\beta},\\
&\tilde\imath^n(v^n,u^n;y_1^n):=\sum_{i=1}^n\tilde\imath(v_i,u_i;y_{1,i}). \label{def:density_uvy}
\end{align}

For smooth presentation, given any $(u^n,v^n,y_1^n,y_2^n)$ and any $t\in\bbR_+$, let
\begin{align}
g_1(t;y_1^n,u^n)&:=\Pr\{\tilde\imath^n(\barV^n;y_1^n|u^n)\geq t\},\label{def:g1}\\
g_2(t;y_1^n,v^n)&:=\Pr\{\tilde\imath^n(\barU^n;y_1^n|v^n)\geq t\},\label{def:g2}\\
g_{1,2}(t,y_2^n)&:=\Pr\{\tilde\imath^n(\barU^n,\barV^n;y_1^n)\geq t\},\label{def:g12}\\
g_3(t,y_2^n)&:=\Pr\{\tilde\imath^n(\barU^n;y_2^n)\geq t\},\label{def:g3}
\end{align}
\begin{align}
\Delta_1&:=2\beta(\beta+\alpha P)\Big(\rmC\Big(\frac{\alpha P}{\beta}\Big)-\log M_1-\log(K_1\sqrt{n})\Big),\label{def:Delta_1}\\
\nn\Delta_3&:=2(1+P)(1+\alpha P)\\*
&\qquad\times\Big(\rmC\Big(\frac{\bar\alpha P}{\alpha P+1}\Big)-\log M_2-\log(K_3\sqrt{n})\Big), \label{def:Delta_3}
\end{align}
where $\barV^n$ and $\barU^n$ are distributed according to $P_{V^n}$ in \eqref{def:PV^n} and $P_{U^n}$ in \eqref{def:PU^n}, respectively.
%\begin{lemma}\label{theo:lemma_g}
%For a large enough $n$,
%\begin{align}
%g_1(t;y_1^n,u^n)&\leq K_1e^{-t}{\sqrt{n}},\\
%g_2(t;y_1^n,v^n)&\leq K_2e^{-t}{\sqrt{n}},\\
%g_{1,2}(t;y_1^n)&\leq K_{1,2}e^{-t}{\sqrt{n}},\\
%g_3(t;y_2^n)&\leq K_3e^{-t}{\sqrt{n}}, \label{def:K_3}
%\end{align}
%where $K_1$, $K_2$ $K_{1,2}$ and $K_{3}$ are finite constants that do not depend on $y_1^n$ or $y_2^n$.
%\end{lemma}

\subsection{Analysis of Error Probability for User 2}\label{sec:phi_2}
We firstly analyze the error probability of user 2. By \eqref{def:density_uy2}, the nearest neighbor rule in \eqref{def:phi2_NN} is equivalent to maximizing $\tilde\imath^n(u^n(w_2);y_2^n)$ over $w_2\in[M_2]$, i.e.,
\begin{align}
\phi_2(Y_2^n)&=\argmin_{w_2\in[M_2]}\|Y_2^n-U^n(w_2)\|^2\\
\nn&=\argmax_{w_2\in[M_2]}n\rmC\left(\frac{\bar\alpha P}{\alpha P+1}\right)+\frac{\|Y_2^n\|^2}{2(P+1)}\\*
&\quad-\frac{\|Y_2^n-U^n(w_2)\|^2}{2(\alpha P+1)}. \label{SEP:phi_2_maxdensity}
\end{align}

Using the random coding union (RCU) bound~\cite[Theorem 16]{polyanskiy2010finite} when the information density is replaced by mismatched information density, the ensemble error probability $\rmP_{\rme,2}^n$ is upper bounded as follows:
\begin{align}
\nn\rmP_{\rme,2}^n&\leq\mathbb{E}\Big[\min\big\{1,\\*
&\quad M_2\Pr\{\tilde\imath^n(\barU^n;Y_2^n) \geq \tilde\imath^n(U^n;Y_2^n)|U^n,Y_2^n\}\big\}\Big], \label{SEP:phi2_useRCU}
\end{align}
where $(\barU^n,U^n,Y_2^n)\sim P_{U^n}(\bar u^n)P_{U^n}(u^n)P_{Y_2^n|U^n}(y_2^n|u^n)$. We can further upper bound $\rmP_{\rme,2}^n$ as follows:
\begin{align}
\rmP_{\rme,2}^n&\leq\mathbb{E}\Big[\min\big\{1,M_2 g_3(\tilde\imath^n(U^n;Y_2^n),Y_2^n)\big\}\Big] \label{SEP:phi2_usedef_g3}\\
&\leq\mathbb{E}\Big[\min\big\{1,M_2 K_3e^{-\tilde\imath^n(U^n;Y_2^n),Y_2^n)}\big\}\Big] \label{SEP:phi2_uselemma_g3}\\
&\leq\Pr\Big\{M_2 K_3e^{-\tilde\imath^n(U^n;Y_2^n),Y_2^n)}\geq\frac{1}{\sqrt{n}}\Big\} +\frac{1}{\sqrt{n}} \label{SEP:phi2_use_min1J}\\
\nn&\leq\frac{1}{\sqrt{n}}+\Pr\Big\{\log M_2-\Big(n\rmC\Big(\frac{\bar\alpha P}{\alpha P+1}\Big)+\frac{\|Y_2^n\|^2}{2(P+1)}\\*
&\qquad\quad-\frac{\|Y_2^n-U^n\|^2}{2(\alpha P+1)}\Big)\geq-\log (K_3\sqrt{n})\Big\}\\
\nn&=\frac{1}{\sqrt{n}}+\Pr\Big\{\frac{\|V^n+Z_2^n\|^2}{2(\alpha P+1)}-\frac{\|U^n+V^n+Z_2^n\|^2}{2(P+1)}\\*
&\qquad\quad\geq \frac{\Delta_3}{2(\alpha P+1)(P+1)}\Big\}, \label{SEP:phi2_use_barZ&Delta}
\end{align}
where \eqref{SEP:phi2_usedef_g3} follows from the definition of $g_3$ in \eqref{def:g3}, \eqref{SEP:phi2_uselemma_g3} follows from~\cite[Lemma 6]{yavas2021gaussian}, and \eqref{SEP:phi2_use_barZ&Delta} follows from $\barZ_2^n:=Z_2^n+V^n$ and \eqref{def:Delta_3}.

Note that user 2 in our ABC setting is the receiver 1 in the non-Gaussian interference channel~\cite[Section III]{scarlett2017mismatch} since we treat the signal transmitted for user 1 as interference. The remaining steps are similar to~\cite[Eq. (108-132)]{scarlett2017mismatch} except that
\begin{enumerate}[i)]
\item $\log M$ is replaced by $\log M_2$,
\item the interference signal is replaced from $\tilde{X}^n$ to $V^n$ and the corresponding per-symbol average power $\tilde{P}$ is replaced by $\alpha P$,
\item the transmitted symbol $X_1^n$ is replaced by $U^n$ and its per-symbol average power is replaced from $P_1$ to $\bar\alpha P$,
\item the noise is replaced from $Z^n$ to $Z_2^n$ and the SINR of the channel is replaced from $\barP$ to $\frac{\bar\alpha P}{\alpha P+1}$.
\end{enumerate}

Choosing
\begin{align}
\log M_2=n\rmC\Big(\frac{\bar\alpha P}{\alpha P+1}\Big)-\sqrt{n}L_2, \label{SEP:chooseM2}
\end{align}
we conclude that
\begin{align}
\rmP_{\rme,2}^n \leq\rmQ\left(\frac{\sqrt{n}L_2-\log(K_2\sqrt{n})}{\sqrt{n \mathrm{V}_2(\bar\alpha P,\alpha P)}}\right)+O\Big(\frac{1}{\sqrt{n}}\Big). \label{SEP:phi2_chooseM2}
\end{align}
Finally, equating the right-hand-side of \eqref{SEP:phi2_chooseM2} to $\varepsilon_2$, solving for $L_2$ and performing a first-order Taylor expansion of $\rmQ^{-1}(\cdot)$ around $\varepsilon_2$, we establish the bound on $L_2$ in Theorem~\ref{theo:SEP}.

\subsection{Analysis of Error Probability for User 1 Under SIC}\label{sec:proof_SEP_1_SIC}
We now upper bound the error probability of user 1 under SIC coding. By the decoding process in Definition~\ref{def:code_SIC}, the error probability of user 1 satisfies
\begin{align}
\rmP_{\rme,1}^n&=\Pr\{\barW_2\neq W_2\;\mathrm{or}\;\hatW_1\neq W_1\} \label{SEP:phi_1_usedef}\\
&=\Pr\{\barW_2\neq W_2\} + \Pr\{\hatW_1\neq W_1,\barW_2=W_2\}. \label{SEP:phi_1_decompose}
\end{align}

Following the similar analysis process of $\rmP_{\rme,2}^n$ and replacing the noise $Z_2^n$ by $Z_1^n$, it follows that
\begin{align}
\nn&\Pr\{\barW_2\neq W_2\}\leq O\Big(\frac{1}{\sqrt{n}}\Big)\\*
&\qquad+\rmQ\left(\frac{n\rmC\Big(\frac{\bar\alpha P}{\alpha P+\beta}\Big)-\log M_2-\log(K_2\sqrt{n})}{\sqrt{n\rmV_2(\bar\alpha P,\alpha P)}}\right)\\
&\qquad\qquad\qquad\;\;\leq O\Big(\frac{1}{\sqrt{n}}\Big), \label{SEP:phi_1_useM2}
\end{align}
where \eqref{SEP:phi_1_useM2} follows from the reasons stated below. Note that we choose $\log M_2=n\rmC(\frac{\bar\alpha P}{\alpha P+1})-\sqrt{n}L_2$ according to Section~\ref{sec:phi_2} and the fact that $\rmC(\frac{\bar\alpha P}{\alpha P+\beta})>\rmC(\frac{\bar\alpha P}{\alpha P+1})$ for $\beta\in(0,1)$. Invoking Chebyshev's inequality and the contents inside $\rmQ(\cdot)$ are with order $O(\sqrt{n})$, the term $\rmQ(\cdot)\leq O(\frac{1}{n})$. An intuitive explanation is that  the channel of user 1 is strictly better than user 2 and thus user 2 can decode the message for user 1 for free. 

Invoking \eqref{SEP:phi_1_decompose} and \eqref{SEP:phi_1_useM2}, the error probability $\rmP_{\rme,1}^n$ can be upper bounded as
\begin{align}
\rmP_{\rme,1}^n&\leq\Pr\{\hatW_1\neq W_1,\barW_2=W_2\}+O\Big(\frac{1}{\sqrt{n}}\Big )\\
&=\Pr\{\hatW_1\neq W_1|\barW_2=W_2\}\Pr\{\barW_2=W_2\}+O\Big(\frac{1}{\sqrt{n}}\Big)\\
&\leq\Pr\{\hatW_1\neq W_1|\barW_2=W_2\}+O\Big(\frac{1}{\sqrt{n}}\Big)
\end{align}
Recall the definition of the nearest neighbor rule in \eqref{def:phi1_NN}, it follows that
\begin{align}
\hatW_1&=\argmin_{w_1\in[M_1]}\|Y_1^n-U^n(\barw_2)-V^n(w_1)\|^2\\
\nn&=\argmax_{w_1\in[M_1]}n\rmC\left(\frac{\alpha P}{\beta}\right)+\frac{\|Y_1^n-U^n(\bar{w}_2)\|^2}{2(\alpha P+\beta)}\\*
&\quad-\frac{\|Y_2^n-U^n(\bar{w}_2)-V^n(w_1))\|^2}{2\beta}\\
&=\argmax_{w_1\in[M_1]}\tilde\imath^n(V^n(w_1);Y_1^n|U^n(\barw_2)). \label{SEP:phi_1_maxdensity}
\end{align}
Proceeding by the similar steps from \eqref{SEP:phi2_useRCU} to \eqref{SEP:phi2_use_barZ&Delta}, we further upper bounded $\rmP_{\rme,1}^n$ as  follows:
\begin{align}
\nn\rmP_{\rme,1}^n
&\leq\frac{1}{\sqrt{n}}+\Pr\Big\{(\alpha P+\beta)\|Z_1^n\|^2-\beta\big(\|Z_1^n\|^2+\|V^n\|^2\\*
&\qquad\quad+2\langle V^n,Z_1^n\rangle\big)\geq \Delta_1\Big\}\\
\nn&=\Pr\Big\{\alpha P\|Z_1^n\|^2-n\alpha\beta P-2\beta\langle V^n,Z_1^n\rangle\geq\Delta_1\Big\}\\*
&\qquad\quad+\frac{1}{\sqrt{n}}\ \label{SEP:phi1_use_powerU},
\end{align}
where \eqref{SEP:phi1_use_powerU} follows since $\|V^n\|^2=n\alpha P$. 

Subsequently, we use the Berry-Esseen Theorem for functions of random vectors (cf. \cite[Proposition 1]{iri2015third},\cite[Theorem 3]{scarlett2017mismatch}). To do so, we write the spherical code in i.i.d. forms. Let $V^n=\sqrt{n\alpha P}\frac{\tilde{V}^n}{\|\tilde{V}^n\|}$, where $\tilde{V}^n\sim\calN(\mathbf{0},\mathbf{I}_n)$. For $i\in[n]$, let $\mathbf{B}_i:=[B_{1,i},B_{2,i},B_{3,i}]$ and
\begin{align}
B_{1,i}:=\beta-Z_{1,i}^2,B_{2,i}:=\sqrt{\alpha P}\tilde{V}_iZ_{1,i},B_{3,i}:=\tilde{V}_i^2-1.
\end{align}
The covariance matrix of $\mathbf{B}_1$ is $\mathrm{Cov}(\mathbf{B}_1)=\diag[\zeta_1-\beta^2\quad \alpha\beta P\quad 2]$.
Define the function
\begin{align}
f_1(b_1,b_2,b_3):=\alpha Pb_1+\frac{2\beta b_2}{\sqrt{1+b_3}}. \label{def:f_1_b123}
\end{align}
One can verify that $-nf_2(\frac{1}{n}\sum_{i=1}^{n}\mathbf{B}_i)=\alpha P\|Z_1^n\|^2-n\beta\alpha P-2\beta\langle V^n,Z_1^n\rangle$ and we can get the Jacobian matrix around $\mathbb{E}\mathbf{B}_1=\mathbf{0}$, i.e., $\mathbf{J_1}=[\alpha P\quad2\beta\quad 0]$.

Note that $\mathbf{B_1},\ldots,\mathbf{B_n}$ are random vectors with mean $\mathbf{0}$ and finite third moment. Then, by Berry-Esseen theorem for functions, we conclude that $(\alpha P\|Z_1^n\|^2-n\beta\alpha P-2\beta\langle V^n,Z_1^n\rangle)/\sqrt{n}$ converges in distribution to a univariate and zero-mean normal random variable with variance $\mathbf{J_1}\mathrm{Cov}(\mathbf{B}_1)\mathbf{J_1}^\mathrm{T}=4\beta^2(\alpha P+\beta)^2\mathrm{V}_1(\alpha P)$, and the converge rate is $O(\frac{1}{\sqrt{n}})$. It follows that
\begin{align}
\rmP_{\rme,1}^n&\leq \rmQ\left(\frac{\Delta_1}{\sqrt{n\mathbf{J}_1\mathrm{Cov}(\mathbf{B}_1)\mathbf{J}_1^\mathrm{T}}} \right)+O\Big(\frac{1}{\sqrt{n}}\Big), \label{SEP:phi1_use_BEtheorem}\\
&=\rmQ\left(\frac{\sqrt{n}L_1-\log(K_1\sqrt{n})}{\sqrt{n \mathrm{V}_1(\alpha P)}}\right)+O\Big(\frac{1}{\sqrt{n}}\Big), \label{SEP:phi1_chooseM1}
\end{align}
where \eqref{SEP:phi1_chooseM1} follows since we choose $\log M_1=n\rmC\Big(\frac{\alpha P}{\beta}\Big)-\sqrt{n}L_1$, and the definition of $\Delta_1$ in \eqref{def:Delta_1}. Finally, equating the right-hand-side of \eqref{SEP:phi1_chooseM1} to $\varepsilon_1$, solving for $L_1$ and performing a first-order Taylor expansion of $\rmQ^{-1}$ about $\varepsilon_1$, we establish the bound on $L_1$ in Theorem~\ref{theo:SEP}.

\subsection{Analysis of Error Probability for User 1 Under JNN}\label{sec:proof_SEP_1_JNN}
To analyze the error probability of user 1 under JNN, we need the following RCU bound.
\begin{lemma}\label{theo:lemma_SEP_RCU_JNN}
There exist an $(n,M_1,M_2)$-JNN code such that the ensemble error probability of user 1 satisfies
\begin{align}
\nn&\rmP_{\rme,1}^n\leq \mathbb{E}\Big[\min\Big\{1,(M_1-1)\Pr\big\{\tilde\imath^n(\barV^n;Y_1^n|U^n)\\*
\nn&\quad\geq\tilde\imath^n(V^n;Y_1^n|U^n)|V^n,U^n,Y_1^n\big\}+(M_2-1)\\*
\nn&\quad\Pr\big\{\tilde\imath^n(\barU^n;Y_1^n|V^n)\geq\tilde\imath^n(U^n;Y_1^n|V^n)|U^n,V^n,Y_1^n\big\}\\ \nn&\quad+(M_1-1)(M_2-1)\Pr\big\{\tilde\imath^n(\barU^n,\barV^n;Y_1^n)\\*
&\quad\geq\tilde\imath^n(U^n,V^n;Y_1^n)|U^n,V^n,Y_1^n\big\}\Big\}\Big],
\end{align}
where the joint distribution of $(U^n,\barU^n,V^n,\barV^n,Y_1^n)$ satisfies
\begin{align}
\nn P_{U^n,\barU^n,V^n,\barV^n,Y_1^n}=&P_{U^n}(u^n)P_{U^n}(\baru^n)P_{V^n}(v^n)P_{V^n}(\barv^n)\\*
&\times P_{Y_1^n|U^nV^n}(y_1^n|u^nv^n).
\end{align}
\end{lemma}
\begin{proof}
In joint decoding, it is naturally to consider the following error events:
\begin{itemize}
\item $\calE_{1,1}:\{\hatW_1\neq W_1, \barW_2=W_2\}$,
\item $\calE_{1,2}:\{\hatW_1=W_1, \barW_2\neq W_2\}$,
\item $\calE_{1,3}:\{\hatW_1\neq W_1, \barW_2\neq W_2\}$.
\end{itemize}

By the union bound, we have that
\begin{align}
\rmP_{\rme,1}^n&=\Pr\{\calE_{1,1}\cup\calE_{1,2}\cup\calE_{1,3}\}\\*
&\leq \Pr\{\calE_{1,1}\}+\Pr\{\calE_{1,2}\}+\Pr\{\calE_{1,3}\}. \label{SEP:phi_1_union}
\end{align}
Now we bound the terms in \eqref{SEP:phi_1_union} separately by exploiting the equivalence between the maximum information density decoder and nearest neighbor decoder. If $\calE_{1,1}$ happens, following from \eqref{SEP:phi_1_maxdensity},
\begin{align}
\phi_1(Y_1^n)=\argmax_{w_1\in[M_1]}\tilde\imath^n(V^n(w_1);Y_1^n|U^n(\barw_2)).
\end{align}
Similarly, under error event $\calE_{1,2}$ and $\calE_{1,3}$, one can show that maximizing the information densities $\tilde\imath^n(U^n(\barw_2);Y_1^n|V^n(w_1))$ and $\tilde\imath^n(V^n(w_1),U^n(\barw_2);Y_1^n)$ are equivalent to minimizing the Euclidean distances, respectively. The rest proof is omitted due to similarity to AWGN channel~\cite[Theorem 16]{polyanskiy2010finite}.
\end{proof}

For subsequent analyses, we need the following notation:
\begin{align}
&\boldsymbol{\tilde{\imath}_\rms}:=\left[
\begin{array}{c}
\tilde\imath^n(V^n;Y_1^n|U^n)\\
\tilde\imath^n(U^n;Y_1^n|V^n)\\
\tilde\imath^n(V^n,U^n;Y_1^n)
\end{array}\right],\;
\mathbf{C_\rms}:=\left[
\begin{array}{c}
\rmC\Big(\frac{\alpha P}{\beta}\Big)\\
\rmC\Big(\frac{\bar\alpha P}{\beta}\Big)\\
\rmC\Big(\frac{P}{\beta}\Big)
\end{array}\right],\\
&\mathbf{M_\rms}:=\left[
\begin{array}{c}
\log (M_1K_1)\\
\log (M_2K_2)\\
\log (M_1M_2K_{1,2})
\end{array}\right],\;
\boldsymbol{\tau_\rms}:=\mathbf{M_\rms}-n\mathbf{C_\rms}, \label{def:M_tau}\\
&\calS_{\rms}:=\big\{\boldsymbol{\tilde{\imath}_\rms}\geq \mathbf{M_\rms}\big\},\qquad\quad
\boldsymbol{\imath_{\rms}}:=[\imath_{\rms,1}\quad\imath_{\rms,2}\quad\imath_{\rms,3}]^T, \label{def:calSs}
\end{align}
\begin{align}
&\imath_{\rms,1}:=\frac{\alpha P(n\beta-\|Z_1^n\|^2)+2\beta\langle V^n,Z_1^n\rangle}{2\beta(\beta+\alpha P)},\\
&\imath_{\rms,2}:=\frac{\bar\alpha P(n\beta-\|Z_1^n\|^2)+2\beta\langle U^n,Z_1^n\rangle}{2\beta(\beta+\bar\alpha P)},\\
&\imath_{\rms,3}:=\frac{P(n\beta-\|Z_1^n\|^2)+2\beta(\langle Z_1^n,U^n+V^n\rangle+\langle U^n,V^n\rangle)}{2\beta(\beta+P)},
\end{align}
where $K_1$, $K_2$ and $K_{1,2}$ are finite constants that do not depend on $y_1^n$ or $y_2^n$. Recall the definitions of $g_1(t;y_1^n,u^n)$, $g_2(t;y_1^n,v^n)$, $g_{1,2}(t;y_2^n)$ and $g_3(t;y_2^n)$ in \eqref{def:g1}-\eqref{def:g3}. For simplicity, we define the following short-hand notation:
\begin{align}
g_1&:=g_1(\tilde\imath^n(V^n;Y_1^n|U^n);Y_1^n,U^n),\\
g_2&:=g_2(\tilde\imath^n(U^n;Y_1^n|V^n);Y_1^n,V^n),\\
g_{1,2}&:=g_{1,2}(\tilde\imath^n(V^n,U^n;Y_1^n);Y_1^n),\\
g_3&:=g_3(\tilde\imath^n(U^n;Y_2^n);Y_2^n).
\end{align}

%Similarly to~\cite[Eq. (46)-(50)]{wang2023achievable} except that
%\begin{itemize}
%\item replacing $\mathbf{\tilde{i}}$ and $\mathbf{i}$ by $\boldsymbol{\tilde{\imath}_{\rms}}$ and $\boldsymbol{\imath_{\rms}}$, respectively,
%\item using the definition of $\mathbf{M_{\rms}}$ and $\boldsymbol{\tau_{\rms}}$ in \eqref{def:M_tau},
%\end{itemize}
%we conclude that
%\begin{align}
%\rmP_{\rme,1}^n\leq 1-\Pr\{\boldsymbol{\imath_{\rms}}\geq\boldsymbol{\tau_{\rms}}\}+O\Big(\frac{1}{\sqrt{n}}\Big). \label{SEP:phi_1_using_yiming}
%\end{align}

With the above definitions, invoking Lemma \ref{theo:lemma_SEP_RCU_JNN}, we upper bound the error probability of user 1 as follows:
\begin{align}
&\nn\rmP_{\rme,1}^n\leq \mathbb{E}\Big[\min\Big\{1,(M_1-1)g_1+(M_2-1)g_2\\*
\nn&\qquad\quad+(M_1-1)(M_2-1)g_{1,2}\}\Big\}1\{\calS_{\rms}^c\}\Big]\\*
\nn&\qquad\quad+\mathbb{E}\Big[\min\Big\{1,(M_1-1)g_1+(M_2-1)g_2\\*
&\qquad\quad+(M_1-1)(M_2-1)g_{1,2}\}\Big\}1\{\calS_{\rms}\}\Big]\\
\nn&\leq \Pr\{\calS_{\rms}^c\}+M_1\mathbb{E} \big[g_11\{\tilde\imath^n(V^n;Y_1^n|U^n)\geq\log(M_1K_1)\}\big]\\*
\nn&\quad+M_2\mathbb{E}\big[g_21\{\tilde\imath^n(U^n;Y_1^n|V^n)\geq\log(M_2K_2)\}\big]\\*
&\quad+M_1M_2\mathbb{E}\big[g_{1,2}1\{\tilde\imath^n(V^n,U^n;Y_1^n)\geq\log(M_1M_2K_{1,2})\}\big] \label{SEP:JNN_use_union_defSs}\\
\nn&\leq\Pr\{\calS_{\rms}^c\}+\frac{1}{\sqrt{n}}\Big(M_1K_1\mathbb{E}\big[\exp\{-\tilde\imath^n(V^n;Y_1^n|U^n)\}\\*
\nn&\qquad1\{\tilde\imath^n(V^n;Y_1^n|U^n)\geq\log(M_1K_1)\}\\*
\nn&\quad+M_2K_2\mathbb{E}\big\{\exp\{-\tilde\imath^n(U^n;Y_1^n|V^n)\}\\*
\nn&\qquad1\{\tilde\imath^n(U^n;Y_1^n|V^n)\geq\log(M_2K_2)\}\\*
\nn&\quad+M_1M_2K_{1,2}\mathbb{E}\big\{\exp\{-\tilde\imath^n(V^n,U^n;Y_1^n)\}\\*
&\qquad1\{\tilde\imath^n(V^n,U^n;Y_1^n)\geq\log(M_1M_2K_{1,2})\}\big]\Big) \label{SEP:JNN_use_Lemma_g}\\
&\leq\Pr\{\calS_{\rms}^c\}+O\Big(\frac{1}{\sqrt{n}}\Big) \label{SEP:JNN_use_neqbound}\\
&\leq 1-\Pr\{\boldsymbol{\imath_{\rms}}\geq\boldsymbol{\tau_{\rms}}\}+O\Big(\frac{1}{\sqrt{n}}\Big). \label{SEP:JNN_use_tildeimath}
\end{align}
where \eqref{SEP:JNN_use_union_defSs} follows from the union bound and the definition of $\calS_{\rms}$ in \eqref{def:calSs}, \eqref{SEP:JNN_use_Lemma_g} follows from~\cite[Lemma 6]{yavas2021gaussian}, \eqref{SEP:JNN_use_neqbound} follows since $\mathbb{E}[\exp\{-a\}1\{a\geq b\}]\leq\exp\{-b\}$ for any variable $a$ and constant $b$, and \eqref{SEP:JNN_use_tildeimath} follows since $\boldsymbol{\tilde{\imath}_\rms} =\boldsymbol{\imath_{\rms}}+n\mathbf{C_\rms}$.

It follows from the definition of $\rmC(P)$ that,
\begin{align}
\rmC\Big(\frac{\bar\alpha P}{\beta}\Big) &>\rmC\Big(\frac{\bar\alpha P}{\alpha P+\beta}\Big) >\rmC\Big(\frac{\bar\alpha P}{\alpha P+1}\Big), \label{SEP:JNN_factM2}\\
\rmC\Big(\frac{P}{\beta}\Big)&=\rmC\Big(\frac{\alpha P}{\beta}\Big) +\rmC\Big(\frac{\bar\alpha P}{\alpha P+\beta}\Big), \label{SEP:JNN_factM1M2}
\end{align}
Recall that in \eqref{SEP:chooseM2}, we choose $M_2$ such that $\log M_2=n\rmC(\frac{\bar\alpha P}{\alpha P+1})-\sqrt{n}L_2$. If we choose $M_1$ such that $\log M_1=n\rmC(\frac{\alpha P}{\beta})-\sqrt{n}L_1$, for sufficiently large $n$, it follows that
\begin{align}
\boldsymbol{\tau_{\rms}}(2)&=n\rmC\Big(\frac{\bar\alpha P}{\alpha P+1}\Big) -n\rmC\Big(\frac{\bar\alpha P}{\beta}\Big)\\
&<0,\label{SEP:taos_2<0}\\
\boldsymbol{\tau_{\rms}}(3)&=n\rmC\Big(\frac{\alpha P}{\beta}\Big)+ n\rmC\Big(\frac{\bar\alpha P}{\alpha P+1}\Big)-n\rmC\Big(\frac{P}{\beta}\Big)\\
&<n\rmC\Big(\frac{\bar\alpha P}{\alpha P+1}\Big)-n\rmC\Big(\frac{\bar\alpha P}{\alpha P+\beta}\Big)\\
&<0. \label{SEP:taos_3<0}
\end{align}
Similar to \eqref{SEP:phi1_use_powerU} to \eqref{def:f_1_b123}, one can verify that $\boldsymbol{\imath_{\rms}}$ converges to $\mathbf{0_3}$. Invoking \eqref{SEP:taos_2<0} and \eqref{SEP:taos_3<0}, we obtain that
\begin{align}
\Pr\{\boldsymbol{\imath_{\rms}}(2,3)\geq \boldsymbol{\tau_{\rms}}(2,3)\}\to 1. \label{SEP:JNN_done_simplify}
\end{align}

Invoking \eqref{SEP:JNN_done_simplify}, we can simplify \eqref{SEP:JNN_use_tildeimath} as follows:
\begin{align}
\rmP_{\rme,1}^n&\leq  1-\Pr\{\boldsymbol{\imath_{\rms}}(1)\geq\boldsymbol{\tau_{\rms}}(1)\}+O\Big(\frac{1}{\sqrt{n}}\Big) \label{SEP:JNN_using_PrAandB}\\
\nn&=\Pr\Big\{\alpha P\|Z_1^n\|^2-n\alpha\beta P-2\beta\langle V^n,Z_1^n\rangle\geq\Delta_1\Big\}\\*
&\qquad\quad+O\Big(\frac{1}{\sqrt{n}}\Big). \label{SEP:phi_1_reduce}
\end{align}
where \eqref{SEP:JNN_using_PrAandB} follows from that $\Pr\{A\cap B\}\geq\Pr\{A\}-\Pr\{B^c\}$, and \eqref{SEP:phi_1_reduce} is identical to \eqref{SEP:phi1_use_powerU} and the {remaing proof steps are} thus omitted.

\section{Proof of Theorem~\ref{theo:JEP}}\label{sec:proof_JEP}
\subsection{Analysis Under SIC}
In the proof of  Theorem~\ref{theo:JEP}, we make extensive use of the definitions of information densities defined in \eqref{def:density_uy2}-\eqref{def:density_uvy}.

Similar to the case under SEP with JNN decoding, we need the following RCU bound under JEP with SIC decoding.
\begin{lemma}\label{theo:lemma_JEP_SIC}
There exist an $(n,M_1,M_2)$-SIC code such that the joint ensemble error probability satisfies
\begin{align}
\nn&\rmP_{\rme,\rmJ}^n\leq \mathbb{E}\Big[\min\Big\{1,(M_1-1)\Pr\{\tilde\imath^n(\barV^n;Y_1^n|U^n)\\*
\nn&\quad\geq\tilde\imath^n(V^n;Y_1^n|U^n)|V^n,U^n,Y_1^n\}+(M_2-1)\\*
\nn&\quad\Pr\{\tilde\imath^n(\barU^n;Y_1^n)\geq\tilde\imath^n(U^n;Y_1^n)|U^n,Y_1^n\}+(M_2-1)\\
&\quad\Pr\{\tilde\imath^n(\barU^n;Y_2^n)\geq\tilde\imath^n(U^n;Y_2^n)|U^n,Y_2^n\}\Big\}\Big],
\end{align}
where the joint distribution of $(U^n,\barU^n,V^n,\barV^n,Y_1^n,Y_2^n)$ satisfies
\begin{align}
\nn &P_{U^n,\barU^n,V^n,\barV^n,Y_1^n,Y_2^n}=P_{U^n}(u^n)P_{U^n}(\baru^n)P_{V^n}(v^n)P_{V^n}(\barv^n)\\*
&\quad\times P_{Y_1^n|U^nV^n}(y_1^n|u^nv^n)\times P_{Y_2^n|U^nV^n}(y_2^n|u^nv^n).
\end{align}
\end{lemma}
The proof of Lemma \ref{theo:lemma_JEP_SIC} is omitted due to similarity to Lemma~\ref{theo:lemma_SEP_RCU_JNN}.

For subsequent analyses of JEP under SIC decoding, we need the following notation:
\begin{align}
&\boldsymbol{\tilde{\imath}_{\rmj\rms}}:=\left[
\begin{array}{c}
\tilde\imath^n(V^n;Y_1^n|U^n)\\
\tilde\imath^n(U^n;Y_1^n)\\
\tilde\imath^n(U^n;Y_2^n)
\end{array}\right],
\mathbf{C_{\boldsymbol{\rmj\rms}}}:=\left[
\begin{array}{c}
\rmC\Big(\frac{\alpha P}{\beta}\Big)\\
\rmC\Big(\frac{\bar\alpha P}{\alpha P+\beta}\Big)\\
\rmC\Big(\frac{\bar\alpha P}{\alpha P+1}\Big)
\end{array}\right]\\
&\mathbf{M_{\boldsymbol{\rmj\rms}}}:=\left[
\begin{array}{c}
\log (M_1K_1)\\
\log (M_2K_2')\\
\log (M_2K_3)
\end{array}\right],\;
\boldsymbol{\tau_{\rmj\rms}}:=\mathbf{M_{\boldsymbol{\rmj\rms}}}- n\mathbf{C_{\boldsymbol{\rmj\rms}}}, \label{def:JEP_M_tau}\\
&\calS_{\boldsymbol{\rmj\rms}}:=\big\{\boldsymbol{\tilde{\imath}_{\rmj\rms}}\geq \mathbf{M_{\boldsymbol{\rmj\rms}}}\big\},\qquad
\boldsymbol{\imath_{\rmj\rms}}:=[\imath_{\rmj\rms,1}\quad\imath_{\rmj\rms,2}\quad\imath_{\rmj\rms,3}]^T, \label{def:calSJs}\\
&\imath_{\rmj\rms,1}:=\imath_{\rms,1},\\
\nn&\imath_{\rmj\rms,2}:=\frac{\bar\alpha P(n\beta-\|Z_1^n\|^2-2\langle V^n,Z_1^n\rangle)}{2(\beta+P)(\beta+\alpha P)}\\
&\qquad\quad+\frac{\langle U^n,Z_1^n\rangle+\langle V^n,U^n\rangle}{\beta+P},\\
\nn&\imath_{\rmj\rms,3}:=\frac{\bar\alpha P(n\beta-\|Z_2^n\|^2-2\langle V^n,Z_2^n\rangle)}{2(1+P)(1+\alpha P)}\\
&\qquad\quad+\frac{\langle U^n,Z_2^n\rangle+\langle V^n,U^n\rangle}{1+P},
\end{align}
where $K_2'$ is finite constant that do not depend on $y_1^n$ or $y_2^n$.

We now upper bound the error probability $\rmP_{\rme,\rmJ}^n$. Invoking Lemma~\ref{theo:lemma_JEP_SIC}, similarly to \eqref{SEP:JNN_use_union_defSs} to \eqref{SEP:JNN_use_tildeimath} except
\begin{itemize}
\item replacing $\boldsymbol{\tilde{\imath}_{\rms}}$ and $\boldsymbol{\imath_{\rms}}$ by $\boldsymbol{\tilde{\imath}_{\rmj\rms}}$ and $\boldsymbol{\imath_{\rmj\rms}}$, respectively,
\item replacing $\mathbf{C}_{\boldsymbol\rms}$ $\mathbf{M_{\boldsymbol\rms}}$, $\boldsymbol{\tau_{\rms}}$ and $\boldsymbol{\calS_{\rms}}$ by $\mathbf{C}_{\boldsymbol{\rmj\rms}}$, $\mathbf{M_{\boldsymbol{\rmj\rms}}}$, $\boldsymbol{\tau_{\rmj\rms}}$ and $\boldsymbol{\calS_{\rmj\rms}}$, respectively.
\end{itemize}
The joint error probability can be upper bounded as
\begin{align}
\rmP_{\rme,\rmJ}^n\leq 1-\Pr\{\boldsymbol{\imath_{\rmj\rms}}\geq\boldsymbol{\tau_{\rmj\rms}}\}+O\Big(\frac{1}{\sqrt{n}}\Big). \label{JEP:phi_1_SIC_using_yiming}
\end{align}

If we choose $\log M_2=n\rmC(\frac{\bar\alpha P}{\alpha P+1})-\sqrt{n}L_2$ and invoke the fact described in \eqref{SEP:JNN_factM2}, it follows that
\begin{align}
\boldsymbol{\tau_{\rmj\rms}}(2)=n\rmC\Big(\frac{\bar\alpha P}{\alpha P+1}\Big) -n\rmC\Big(\frac{\bar\alpha P}{\alpha P+\beta}\Big)< 0.
\end{align}
Similar to \eqref{SEP:phi1_use_powerU} to \eqref{def:f_1_b123}, one can verify that $\boldsymbol{\imath_{\rmj\rms}}$ converges to $\mathbf{0_3}$. Thus, we obtain that
\begin{align}
\Pr\{\boldsymbol{\imath_{\rmj\rms}}(2)\geq \boldsymbol{\tau_{\rmj\rms}}(2)\}\to 1.
\end{align}
Thus, we can simplify \eqref{JEP:phi_1_SIC_using_yiming} as follows:
\begin{align}
\rmP_{\rme,\rmJ}^n\leq 1-\Pr\{\boldsymbol{\imath_{\rmj\rms}}(1,3)\geq\boldsymbol{\tau_{\rmj\rms}}(1,3)\} +O\Big(\frac{1}{\sqrt{n}}\Big). \label{JEP:SIC_simplify}
\end{align}
We subsequently lower bound the term $\Pr\{\boldsymbol{\imath_{\rmj\rms}}(1,3)\geq\boldsymbol{\tau_{\rmj\rms}}(1,3)\}$. We can denote the spherical codewords by $V^n=\sqrt{n\alpha P}\frac{\tilde{V}^n}{\|\tilde{V}^n\|}$ and $U^n=\sqrt{n\bar\alpha P}\frac{\tilde{U}^n}{\|\tilde{U}^n\|}$, where $\tilde{V}^n\sim\calN(\mathbf{0},\mathbf{I}_n)$ and $\tilde{U}^n\sim\calN(\mathbf{0},\mathbf{I}_n)$. Furthermore, for all $i\in[n]$, let $\mathbf{A}_i:=(A_{1,i},\ldots,A_{8,i})$, where
\begin{align}
\begin{array}{ll}
A_{1,i}:=\beta-Z_{1,i}^2, & A_{2,i}:=\sqrt{\alpha P}\tilde{V}_iZ_{1,i}, \\ A_{3,i}:=\tilde{V}_i^2-1,& A_{4,i}:=1-Z_{2,i}^2,\\
A_{5,i}:=\sqrt{\bar\alpha P}\tilde{U}_iZ_{2,i},& A_{6,i}:=\tilde{U}_i^2-1,\\
A_{7,i}:=\sqrt{\bar\alpha \alpha P^2}\tilde{V}_i\tilde{U}_i, & A_{8,i}:=\sqrt{\alpha P}\tilde{V}_iZ_{2,i}.
\end{array}
\end{align}
We next define the following function  $\mathbf{f_{\rmJ}}:=[f_{\rmj,1}(\mathbf{a}),f_{\rmj,2}(\mathbf{a})]^T$, where
\begin{align}
f_{\rmJ,1}(\mathbf{a}):=&\frac{1}{2\beta(\beta+\alpha P)}\Big(\alpha Pa_1+\frac{2\beta a_2}{\sqrt{1+a_3}}\Big),\\
\nn f_{\rmJ,2}(\mathbf{a}):=&\frac{1}{2(1+P)(1+\alpha P)}\bigg(\bar\alpha P\Big(a_4-\frac{2a_8}{\sqrt{1+a_3}}\Big)\\
&+2(1+\alpha P)\Big(\frac{a_5}{\sqrt{1+a_6}}+\frac{a_7}{\sqrt{1+a_6}\sqrt{1+a_3}}\Big)\bigg).
\end{align}
One can verify that $\mathbf{f}_\rmJ(\frac{1}{n}\sum_{i=1}^{n}\mathbf{A}_i)=\frac{1}{n} \boldsymbol{\imath_{\rmj\rms}}(1,3)$.

The variances of each of the random variables $A_{1,i},\ldots,A_{8,i}$ and the Jacobian of $\mathbf{f}_{\rmJ}$ around $\mathbb{E}[\mathbf{A}_1]=\mathbf{0}_8$ with respect to each of its entries are as follows:
\begin{align}
\begin{array}{cll}
\mathrm{Random\;variable} & \mathrm{Variance} & \mathrm{Jacobian\;entry}\\
A_1 & \zeta_1-\beta^2 & k_1\alpha P\\
A_2 & \alpha\beta P & 2k_1\beta\\
A_3 & 2 & 0\\
A_4 & \zeta_2-1 & k_2\alpha P\\
A_5 & \bar\alpha P & 2K_2(1+\alpha P)\\
A_6 & 2 & 0\\
A_7 & \alpha\bar\alpha P^2 & 2k_2(1+\alpha P)\\
A_8 & \alpha P & -2\bar\alpha P.\\
\end{array}\label{JEP:SIC_variance}
\end{align}
The covariance matrix $\mathbf{V}_{\rmJ}$ of the random vector $\mathbf{A}$ is diagonal and the Jacobian matrix can be denoted by $\mathbf{J}_{\rmJ}=(J_{\rmJ,1},J_{\rmJ,2})^T$, where
\begin{align}
J_{\rmJ,1}&:=[k_1\alpha P\;2k_1\beta\;0\;0\;0\;0\;0\;0],\\
J_{\rmJ,2}&:=[0\;0\;0\;k_2\alpha P\;2K_2(1+\alpha P)\;0\;2k_2(1+\alpha P)\; -2\bar\alpha P]. \label{JEP:SIC_JJ2}
\end{align}
Combining \eqref{JEP:SIC_variance}-\eqref{JEP:SIC_JJ2}, we have
\begin{align}
\mathbf{J}_{\rmJ}\mathbf{V}_{\rmJ}\mathbf{J}_{\rmJ}^T&=
\begin{bmatrix}
\rmV_1(\alpha P) & 0 \\
0 & \rmV_2(\bar\alpha P,\alpha P)
\end{bmatrix}\\
&=\mathbf{V}\big(\alpha,P\big).
\end{align}
where $\rmV_1(\alpha P)$, $\rmV_2(\bar\alpha P,\alpha P)$ and $\mathbf{V}\big(\alpha,P\big)$ are defined after \eqref{def:V2()}.
%\begin{align}
%\rmV_{\rmJ,1}&=\frac{\alpha^2P^2(\zeta_1-\beta^2)+4\alpha P\beta^3}{4\beta^2(\alpha P+\beta)^2},\\
%\rmV_{\rmJ,2}&=\frac{\bar\alpha^2P^2(\zeta_2-1+4\bar\alpha P)+4\bar\alpha P(\bar\alpha P +1)^3}{4(\bar\alpha P+1)^2(P+1)^2}.
%\end{align}

By the Berry-Esseen Theorem for functions of random vectors~\cite[Proposition 1]{molavianjazi2015second}, there exists a finite positive constant $B$ such that
\begin{align}
\nn&\left|\Pr \left\{\mathbf{f}_{\rmJ}\left(\frac{1}{n}\sum_{i=1}^{n} \mathbf{A}_{i}\right)\geq\frac{\boldsymbol{\tau_{\rmj\rms}}(1,3)}{n}\right\} -\Pr\left\{\mathbf{S}_2\geq \frac{\boldsymbol{\tau_{\rmj\rms}}(1,3)}{\sqrt{n}}\right\}\right|\\*
&\quad \leq\frac{B}{n^{\frac{1}{4}}},
\end{align}
where $\mathbf{S}_2\sim \calN\big(\mathbf{0}_2,\mathbf{V}(\alpha,P)\big)$. Invoking the fact that $\mathbf{f}_\rmJ(\frac{1}{n}\sum_{i=1}^{n}\mathbf{A}_i)=\frac{1}{n} \boldsymbol{\imath_{\rmj\rms}}(1,3)$, it follows that
\begin{align}
\nn&\Pr\{\boldsymbol{\imath_{\rmj\rms}}(1,3)\geq\boldsymbol{\tau_{\rmj\rms}}(1,3)\}- \Pr\Big\{\mathbf{S}_2 \geq\frac{\boldsymbol{\tau_{\rmj\rms}}(1,3)}{\sqrt{n}}\Big\}\\*
&\qquad\geq O\Big(n^{-\frac{1}{4}}\Big). \label{JEP:SIC_using_A=tau}
\end{align}
Combining \eqref{JEP:SIC_simplify} and \eqref{JEP:SIC_using_A=tau}, the joint error probability can be finally upper bounded by
\begin{align}
\rmP_{\rme,\rmJ}^n&\leq 1-\Pr\Big\{\mathbf{S}_2\geq\frac{\boldsymbol{\tau_{\rmj\rms}}(1,3)}{\sqrt{n}}\Big\} +O\Big(n^{-\frac{1}{4}}\Big)\\
&= 1-\Pr\Big\{\mathbf{S}_2\leq-\frac{\boldsymbol{\tau_{\rmj\rms}}(1,3)}{\sqrt{n}}\Big\} +O\Big(n^{-\frac{1}{4}}\Big).\label{JEP:SIC_combine_decomp_BE}
\end{align}
Recall the definition of $\rmQ_{\mathrm{inv}}(\mathbf{V},\varepsilon)$ in \eqref{def:Qinv}, let
\begin{align}
-\frac{\boldsymbol{\tau_{\rmj\rms}}(1,3)}{\sqrt{n}}\in\rmQ_{\mathrm{inv}} \left(\mathbf{V}(\alpha,P), \varepsilon-O\Big(n^{-\frac{1}{4}}\Big)\right).
\end{align}
Recall the definition of achievable pair $(L_1,L_2)$ in Def. \ref{def:JEP_achL1L2}, the proof is completed by applying the Taylor expansion of $\rmQ_{\mathrm{inv}} (\cdot)$ around $\varepsilon$. 

\subsection{Analysis Under JNN}
To prove the results under JNN decoding, we need the following RCU bound.
\begin{lemma}\label{theo:lemma_JEP_JNN}
There exist an $(n,M_1,M_2)$-JNN code such that the joint ensemble error probability satisfies
\begin{align}
\nn&\rmP_{\rme,\rmJ}^n\leq \mathbb{E}\Big[\min\Big\{1,(M_1-1)\Pr\big\{\tilde\imath^n(\barV^n;Y_1^n|U^n)\\*
\nn&\quad\geq\tilde\imath^n(V^n;Y_1^n|U^n)|V^n,U^n,Y_1^n\big\}+(M_2-1)\\*
\nn&\quad\Pr\big\{\tilde\imath^n(\barU^n;Y_1^n|V^n)\geq\tilde\imath^n(U^n;Y_1^n|V^n)|U^n,V^n,Y_1^n\big\}\\ \nn&\quad+(M_1-1)(M_2-1)\Pr\big\{\tilde\imath^n(\barU^n,\barV^n;Y_1^n)\\*
&\quad\geq\tilde\imath^n(U^n,V^n;Y_1^n)|U^n,V^n,Y_1^n\big\}+(M_2-1)\\
&\quad\Pr\big\{\tilde\imath^n(\barU^n;Y_2^n)\geq\tilde\imath^n(U^n;Y_2^n)|U^n,Y_2^n\big\} \Big\}\Big],
\end{align}
where the joint distribution of $(U^n,\barU^n,V^n,\barV^n,Y_1^n,Y_2^n)$ satisfies
\begin{align}
\nn &P_{U^n,\barU^n,V^n,\barV^n,Y_1^n,Y_2^n}=P_{U^n}(u^n)P_{U^n}(\baru^n)P_{V^n}(v^n)P_{V^n}(\barv^n)\\*
&\quad\times P_{Y_1^n|U^nV^n}(y_1^n|u^nv^n)\times P_{Y_2^n|U^nV^n}(y_2^n|u^nv^n).
\end{align}
\end{lemma}
The proof of Lemma~\ref{theo:lemma_JEP_JNN} is omitted due to similarity to  Lemma~\ref{theo:lemma_SEP_RCU_JNN}.
%\begin{proof}
%The error events of user 1 is similar to the case under SEP with JNN decoding, and we need to additionally consider the error probability of user 2. Hence, there are four error events:
%\begin{itemize}
%\item $\calE_{1,1}:\{\hatW_1\neq W_1, \barW_2=W_2\}$,
%\item $\calE_{1,2}:\{\hatW_1=W_1, \barW_2\neq W_2\}$,
%\item $\calE_{1,3}:\{\hatW_1\neq W_1, \barW_2\neq W_2\}$,
%\item $\calE_2:\{\hatW_2\neq W_2\}$.
%\end{itemize}
%The rest of proof is similar to Lemma \ref{theo:lemma_SEP_RCU_JNN} and thus is omitted.
%\end{proof}

For the rest of analyses, we need the following notation:
\begin{align}
&\boldsymbol{\tilde{\imath}_{\rmj\rmj}}:=\left[
\begin{array}{c}
\tilde\imath^n(V^n;Y_1^n|U^n)\\
\tilde\imath^n(U^n;Y_1^n|V^n)\\
\tilde\imath^n(V^n,U^n;Y_1^n)\\
\tilde\imath^n(U^n;Y_2^n)
\end{array}\right],\;
\mathbf{C_{\boldsymbol{\rmj\rmj}}}:=\left[
\begin{array}{c}
\rmC\Big(\frac{\alpha P}{\beta}\Big)\\
\rmC\Big(\frac{\bar\alpha P}{\beta}\Big)\\
\rmC\Big(\frac{P}{\beta}\Big)\\
\rmC\Big(\frac{\bar\alpha P}{\alpha P+1}\Big)
\end{array}\right],\\
&\mathbf{M_{\boldsymbol{\rmj\rmj}}}:=\left[
\begin{array}{c}
\log (M_1K_1)\\
\log (M_2K_2)\\
\log (M_1M_2K_{1,2})\\
\log (M_2K_3)
\end{array}\right],\;
\boldsymbol{\tau_{\rmj\rmj}}:=\mathbf{M_{\boldsymbol{\rmj\rmj}}}- n\mathbf{C_{\boldsymbol{\rmj\rmj}}}, \label{def:JEP_JNN_M_tau}\\
&\calS_{\boldsymbol{\rmj\rmj}}:=\big\{\boldsymbol{\tilde{\imath}_{\rmj\rmj}}\geq \mathbf{M_{\boldsymbol{\rmj\rmj}}}\big\},\qquad
\boldsymbol{\imath_{\rmj\rmj}}:=[\imath_{\rmj\rmj,1}\quad\imath_{\rmj\rmj,2} \quad\imath_{\rmj\rmj,3}\quad\imath_{\rmj\rmj,4}]^T, \label{def:calSJJ}
\end{align}
\begin{align}
&\imath_{\rmj\rmj,1}:=\frac{\alpha P(n\beta-\|Z_1^n\|^2)+2\beta\langle V^n,Z_1^n\rangle}{2\beta(\beta+\alpha P)},\\
&\imath_{\rmj\rmj,2}:=\frac{\bar\alpha P(n\beta-\|Z_1^n\|^2)+2\beta\langle U^n,Z_1^n\rangle}{2\beta(\beta+\bar\alpha P)},\\
&\imath_{\rmj\rmj,3}:=\frac{P(n\beta-\|Z_1^n\|^2)+2\beta(\langle Z_1^n,U^n+V^n\rangle+\langle U^n,V^n\rangle)}{2\beta(\beta+P)},\\
\nn&\imath_{\rmj\rmj,4}:=\frac{\bar\alpha P(n\beta-\|Z_2^n\|^2-2\langle V^n,Z_2^n\rangle)}{2(1+P)(1+\alpha P)}\\
&\qquad\quad+\frac{\langle U^n,Z_2^n\rangle+\langle V^n,U^n\rangle}{1+P},
\end{align}
Invoking Lemma~\ref{theo:lemma_JEP_JNN}, similarly to \eqref{SEP:JNN_use_union_defSs} to \eqref{SEP:JNN_use_tildeimath} except that
\begin{itemize}
\item replacing $\boldsymbol{\tilde{\imath}_{\rms}}$ and $\boldsymbol{\imath_{\rms}}$ by $\boldsymbol{\tilde{\imath}_{\rmj\rmj}}$ and $\boldsymbol{\imath_{\rmj\rmj}}$, respectively,
\item replacing $\mathbf{C}_{\boldsymbol\rms}$ $\mathbf{M_{\boldsymbol\rms}}$, $\boldsymbol{\tau_{\rms}}$ and $\boldsymbol{\calS_{\rms}}$ by $\mathbf{C}_{\boldsymbol{\rmj\rmj}}$, $\mathbf{M_{\boldsymbol{\rmj\rmj}}}$, $\boldsymbol{\tau_{\rmj\rmj}}$ and $\boldsymbol{\calS_{\rmj\rmj}}$, respectively,
\end{itemize}
we upper bound the joint error probability as
\begin{align}
\rmP_{\rme,\rmJ}^n\leq 1-\Pr\{\boldsymbol{\imath_{\rmj\rmj}}\geq\boldsymbol{\tau_{\rmj\rmj}}\}+O\Big(\frac{1}{\sqrt{n}}\Big). \label{JEP:phi_1_JNN_using_yiming}
\end{align}
For any $\beta\in(0,1)$, recall the fact in \eqref{SEP:JNN_factM2} and \eqref{SEP:JNN_factM1M2}. If we choose $M_1$ and $M_2$ such that $\log M_1=n\rmC(\frac{\alpha P}{\beta})-\sqrt{n}L_1$ and $\log M_2=n\rmC(\frac{\bar\alpha P}{\alpha P+1})-\sqrt{n}L_2$, for sufficiently large $n$, it follows that
\begin{align}
\boldsymbol{\tau_{\rmj\rmj}}(2)&=n\rmC\Big(\frac{\bar\alpha P}{\alpha P+1}\Big) -n\rmC\Big(\frac{\bar\alpha P}{\beta}\Big)<0,\label{JEP:JNN_factM2}\\
\boldsymbol{\tau_{\rmj\rmj}}(3)&=n\rmC\Big(\frac{\alpha P}{\beta}\Big)+ n\rmC\Big(\frac{\bar\alpha P}{\alpha P+1}\Big)-n\rmC\Big(\frac{P}{\beta}\Big)\\
&<n\rmC\Big(\frac{\bar\alpha P}{\alpha P+1}\Big)-n\rmC\Big(\frac{\bar\alpha P}{\alpha P+\beta}\Big)<0. \label{JEP:JNN_donesimplify}
\end{align}
Similar to \eqref{SEP:phi1_use_powerU} to \eqref{def:f_1_b123}, one can verify that $\boldsymbol{\imath_{\rms}}$ converges to $\mathbf{0_3}$. Thus, we obtain that
\begin{align}
\Pr\{\boldsymbol{\imath_{\rmj\rmj}}(2,3)\geq \boldsymbol{\tau_{\rmj\rmj}}(2,3)\}\to 1. \label{JEP:JNN_done_simplify}
\end{align}
Invoking \eqref{JEP:JNN_done_simplify}, similar to \eqref{SEP:JNN_using_PrAandB},  we can simplify \eqref{JEP:phi_1_JNN_using_yiming} as follows:
\begin{align}
\rmP_{\rme,\rmJ}^n\leq 1-\Pr\{\boldsymbol{\imath_{\rmj\rmj}}(1,4)\geq\boldsymbol{\tau_{\rmj\rmj}}(1,4)\} +O\Big(\frac{1}{\sqrt{n}}\Big). \label{JEP:JNN_simplify}
\end{align}
Since \eqref{JEP:JNN_simplify} is identical to \eqref{JEP:SIC_simplify}, the rest of proof is thus omitted.

\section{Proof Sketch of Theorem~\ref{theo:fading_SEP}}\label{sec:proof_fading}
The proof of Theorem~\ref{theo:fading_SEP} based on the similar steps in Section~\ref{sec:proof_sketch_SEP}. However, the existing of the fading parameter leads to some additional constraints. Thus, for the smooth presentation, we provide the analysis of $\rmP_{\rme,2}^n$ and we only provide the sketch for the analysis of $\rmP_{\rme,1}^n$.
\subsection{Analysis of Error Probability for User 2}
Similar to the case without fading, ``mismatched fading'' information densities are critical to our analysis. Specifically, given any $(u,y_2,h_2)\in\bbR^3$, let
\begin{align}
\nn&\tilde\imath(u;y_2|h_2)\\*
&:=\log\frac{\calN(y_2;h_2u,1+h_2^2\alpha P)}{\calN(y;0,1+h_2^2P)}\\
&=\rmC\left(\frac{h_2^2\bar\alpha P}{h_2^2\alpha P+1}\right)+\frac{y_2^2}{2(h_2^2P+1)}-\frac{(y_2-h_2u)^2}{2(h_2^2\alpha P+1)}, \label{def:fading_density_uy2}
\end{align}
which is the information density of the quasi-static fading Gaussian channel $\calN(y_2;h_2u,1+h_2^2\alpha P)$. For all $n\in\bbN$, we define the following $n$-letter version of the mismatched fading information density
\begin{align}
\tilde\imath^n(u^n;y_2^n|h_2):=\sum_{i=1}^n\tilde\imath(u_i;y_{2,i}|h_2).
\end{align}

Firstly, we analyze the error probability of user 2. Similar to the case without fading, the nearest neighbor rule is equivalent to maximizing $\tilde\imath^n(u^n(w_2);y_2^n|h_2)$ over $w_2\in[M_2]$, i.e.,
\begin{align}
\phi_2(Y_2^n)&=\argmin_{w_2\in[M_2]}\|Y_2^n-H_2U^n(w_2)\|^2\\
\nn&=\argmax_{w_2\in[M_2]}n\rmC_2(H_2,\alpha P,\bar\alpha P)+\frac{\|Y_2^n\|^2}{2(H_2^2P+1)}\\*
&\quad-\frac{\|Y_2^n-H_2U^n(w_2)\|^2}{2(H_2^2\alpha P+1)}. \label{SEP:fading_phi_2_maxdensity}
\end{align}

Proceeding with the similar steps in Section~\ref{sec:phi_2}, the error probability of user 2 can be upper bounded as
\begin{align}
\nn&\rmP_{\rme,2}^n\\*
\nn&\leq\mathbb{E}_{H_2}\bigg[\Pr\Big\{\frac{\|H_2V^n+Z_2^n\|^2}{2(H_2^2\alpha P+1)}-\frac{\|H_2U^n+H_2V^n+Z_2^n\|^2}{2(H_2^2P+1)}\\*
&\qquad\quad\geq \frac{\Delta_f}{2(H_2^2\alpha P+1)(H_2^2P+1)}\Big\}+\frac{1}{\sqrt{n}}\bigg], \label{SEP:fading_phi2_use_barZ&Delta}
\end{align}
where
\begin{align}
\nn\Delta_f&:=\bigg(\rmC_2(H_2,\alpha P,\bar\alpha P)-\log M_2-\log(K_3\sqrt{n})\bigg)\\*
&\qquad\times2(1+H_2^2P)(1+H_2^2\alpha P). \label{def:Delta_f}
\end{align}
Then, by multiplying $2(H_2^2\alpha P+1)(H_2^2P+1)$ at the both side of the probability term, the left-hand side of the probability term can be expanded as
\begin{align}
\nn&2(1+H_2^2P)\|H_2V^n+Z_2^n\|^2-(1+H_2^2\alpha P)\\*
\nn&\quad\times\|H_2U^n+H_2V^n+Z_2^n\|^2\\
\nn&=H_2^2\bar\alpha P\big(-n+\|Z_2^n\|^2+2\langle H_2V^n,Z_2^n\rangle\big) -2(1+H_2^2\alpha P)\\*
&\quad\times\big(\langle H_2U^n,H_2V^n\rangle+\langle H_2U^n,Z_2^n\rangle\big), \label{SEP:fading_phi2_expend}
\end{align}
since $\|U^n\|^2=n\bar\alpha P$ and $\|V^n\|^2=n\alpha P$.

In order to apply the functional Berry-Esseen Theorem, we write \eqref{SEP:fading_phi2_expend} in terms of a function of random variables. Recall that $V^n=\sqrt{n\alpha P}\frac{\tilde{V}^n}{\|\tilde{V}^n\|}$ and $U^n=\sqrt{n\bar\alpha P}\frac{\tilde{U}^n}{\|\tilde{U}^n\|}$, where $\tilde{V}^n\sim\calN(\mathbf{0},\mathbf{I}_n)$ and $\tilde{U}^n\sim\calN(\mathbf{0},\mathbf{I}_n)$. For $i\in[n]$, let 
\begin{align}
&A_{1,i}^F:=H_2^2\bar\alpha P(1-Z_2^2),\\
&A_{2,i}^F:=\sqrt{\bar\alpha P}(1+H_2^2\alpha P)H_2\tilde{U}_iZ_{2,i},\\
&A_{3,i}^F:=\tilde{U}_i^2-1,\\
&A_{4,i}^F:=\sqrt{\alpha\bar\alpha P^2}(1+H_2^2\alpha P)H_2^2\tilde{U}_i\tilde{V}_i,\\
&A_{5,i}^F:=\sqrt{\alpha P}(1+H_2^2\alpha P)H_2\tilde{V}_iZ_{2,i},\\
&A_{6,i}^F:=\tilde{V}_i^2-1.
\end{align}

One can then verify that $\mathbf{A}_i^F$ is zero-mean. With the assumption of $H_2$ in \eqref{def:fading_SEP_assump_H}, we conclude that $\mathbb{E}[\|\mathbf{A}_i^F\|^3]<\infty$, which is a critical condition for applying Berry-Esseen Theorem.

We then define the function
\begin{align}
\nn f_F(a_1^F,\dots,a_6^F)&:=a_1^F-\frac{a_5^F}{\sqrt{1+a_6^F}}+2\left( \frac{a_2^F}{\sqrt{1+a_3^F}}\right.\\*
&\quad\left.+\frac{a_4^F}{\sqrt{1+a_3^F}\sqrt{1+a_6^F}}\right).
\end{align}
The right-hand side of \eqref{SEP:fading_phi2_expend} can be denoted by $-nf_F\left(\frac{1}{n}\sum_{i=1}^n\mathbf{A}_i^F\right)$.

We then apply the Berry-Esseen Theorem for functions in \eqref{SEP:fading_phi2_use_barZ&Delta} and invoke \eqref{SEP:fading_phi2_expend}, it follows that
\begin{align}
\nn&\rmP_{\rme,2}^n\leq O\Big(\frac{1}{\sqrt{n}}\Big)\\
&+\mathbb{E}_{H_2}\left[\rmQ\left(\frac{n\rmC_2(H_2,\alpha P,\bar\alpha P)-\log M_2-\log(K_3\sqrt{n})}{\sqrt{n \mathrm{V}_2(H_2^2\bar\alpha P,H_2^2\alpha P)}}\right)\right].
\end{align}
We then perform a first-order Taylor expansion of $\rmQ(\cdot)$, by noting the fact that $\rmV_2(H_2^2\bar\alpha P,H_2^2\alpha P)=O(1)$, we complete the proof of \eqref{theo:fading_Pe2}.

\subsection{Analysis Sketch of User 1}
The proof of User 1 under SIC and JNN decoding scheme follows from the similar steps in Section~\ref{sec:proof_SEP_1_SIC} and Section~\ref{sec:proof_SEP_1_JNN}, respectively. Specifically, the proof steps are different in the following aspects:
\begin{itemize}
\item The analyses are based on ``mismatched fading'' information densities;
\item Applying the moment constraint $\mathbb{E}[H_1^9]<\infty$ when using the Berry-Esseen Theorem.
\end{itemize}

\section{Conclusion}\label{sec:conclusion}
We derived achievable second-order asymptotics for mismatched asymmetric broadcast channel, where the transmitter uses spherical codebooks and the decoders apply nearest neighbor decoder and its variations. Our results show that using SIC and JNN decoding for the strong user leads to identical second-order performance, which advocates SIC for practical use since it has much lower computational complexity. Our results also show that the JEP criterion is superior than the SEP criterion by having a larger second-order achievable region. We also generalized our results to the case with quasi-static fading and verified that outage capacity region is an accurate notation even at finite blocklength. 

There are several avenues for future research. Firstly, we only derived achievability results in this paper. However, in order to check whether our results are ensemble tight, it is worthwhile to derive the corresponding ensemble converse results for the mismatched ABC. Towards this goal, the techniques in \cite{scarlett2017mismatch,zhou2018refined,zhou2019jscc} might be helpful.  Secondly, for the case with quasi-static fading, we assumed that perfect CSI is available at decoders. However, in practice, perfect CSI might not be available. Thus, it is worthwhile to generalize our results to the case with knowledge of CSI statistics or even with no CSI at both receivers, potentially using tools from~\cite{yang2014quasi,Li2014outageDNOMA,wu2014BCtsp}. Finally, we only considered channel coding in this paper. For practical applications, joint source-channel coding (JSCC) is needed. Towards this goal, it is of great value to study a mismatched JSCC setting for multiple users. To do so, one might combine the results in this paper with the source coding counterpart of successive refinement~\cite{rimoldi1994,wu2023TIT} to construct a JSCC code via unequal error protection~\cite{zhou2019jscc,wang2011dispersion}.

%\newpage
\bibliographystyle{IEEEtran}
\bibliography{IEEEfull_fei}

\end{document}

%% file: preamble.tex
\usepackage{soul}
\usepackage[mathscr]{eucal}
\usepackage[cmex10]{amsmath}
\usepackage{epsfig,epsf,psfrag}
\usepackage{amssymb,amsmath,amsthm,amsfonts,latexsym}
\usepackage{amsmath,graphicx,bm,xcolor,url,overpic}
\usepackage{array}%array and tabular environments
\usepackage{verbatim}
\usepackage{bm}
\usepackage{algorithmic}
\usepackage{algorithm}
\usepackage{verbatim}
\usepackage{textcomp}
\usepackage{mathrsfs}
\usepackage{epstopdf}

\newcommand{\openone}{\leavevmode\hbox{\small1\normalsize\kern-.33em1}}

\catcode`~=11 \def\UrlSpecials{\do\~{\kern -.15em\lower .7ex\hbox{~}\kern .04em}} \catcode`~=13

\allowdisplaybreaks[4]

\newcommand{\nn}{\nonumber}

\newcommand{\calE}{\mathcal{E}}

\newcommand{\calH}{\mathcal{H}}

\newcommand{\calL}{\mathcal{L}}

\newcommand{\calN}{\mathcal{N}}

\newcommand{\calR}{\mathcal{R}}
\newcommand{\calS}{\mathcal{S}}

\newcommand{\bV}{\mathbf{V}}

\newcommand{\rmA}{\mathrm{A}}

\newcommand{\rmC}{\mathrm{C}}

\newcommand{\rme}{\mathrm{e}}

\newcommand{\rmj}{\mathrm{j}}
\newcommand{\rmJ}{\mathrm{J}}

\newcommand{\rmP}{\mathrm{P}}

\newcommand{\rmQ}{\mathrm{Q}}

\newcommand{\rms}{\mathrm{s}}

\newcommand{\rmV}{\mathrm{V}}

\newcommand{\bbN}{\mathbb{N}}

\newcommand{\bbR}{\mathbb{R}}

\DeclareMathAlphabet{\mathbsf}{OT1}{cmss}{bx}{n}
\DeclareMathAlphabet{\mathssf}{OT1}{cmss}{m}{sl}% slanted sans serif

\DeclareSymbolFont{bsfletters}{OT1}{cmss}{bx}{n}
\DeclareSymbolFont{ssfletters}{OT1}{cmss}{m}{n}
\DeclareMathSymbol{\bsfGamma}{0}{bsfletters}{'000}
\DeclareMathSymbol{\ssfGamma}{0}{ssfletters}{'000}
\DeclareMathSymbol{\bsfDelta}{0}{bsfletters}{'001}
\DeclareMathSymbol{\ssfDelta}{0}{ssfletters}{'001}
\DeclareMathSymbol{\bsfTheta}{0}{bsfletters}{'002}
\DeclareMathSymbol{\ssfTheta}{0}{ssfletters}{'002}
\DeclareMathSymbol{\bsfLambda}{0}{bsfletters}{'003}
\DeclareMathSymbol{\ssfLambda}{0}{ssfletters}{'003}
\DeclareMathSymbol{\bsfXi}{0}{bsfletters}{'004}
\DeclareMathSymbol{\ssfXi}{0}{ssfletters}{'004}
\DeclareMathSymbol{\bsfPi}{0}{bsfletters}{'005}
\DeclareMathSymbol{\ssfPi}{0}{ssfletters}{'005}
\DeclareMathSymbol{\bsfSigma}{0}{bsfletters}{'006}
\DeclareMathSymbol{\ssfSigma}{0}{ssfletters}{'006}
\DeclareMathSymbol{\bsfUpsilon}{0}{bsfletters}{'007}
\DeclareMathSymbol{\ssfUpsilon}{0}{ssfletters}{'007}
\DeclareMathSymbol{\bsfPhi}{0}{bsfletters}{'010}
\DeclareMathSymbol{\ssfPhi}{0}{ssfletters}{'010}
\DeclareMathSymbol{\bsfPsi}{0}{bsfletters}{'011}
\DeclareMathSymbol{\ssfPsi}{0}{ssfletters}{'011}
\DeclareMathSymbol{\bsfOmega}{0}{bsfletters}{'012}
\DeclareMathSymbol{\ssfOmega}{0}{ssfletters}{'012}

\newcommand{\hatW}{\hat{W}}

\newcommand{\baru}{\bar{u}}
\newcommand{\barv}{\bar{v}}
\newcommand{\barw}{\bar{w}}

\newcommand{\barP}{\bar{P}}

\newcommand{\barU}{\bar{U}}
\newcommand{\barV}{\bar{V}}
\newcommand{\barW}{\bar{W}}

\newcommand{\barZ}{\bar{Z}}

\DeclareMathOperator*{\argmax}{arg\,max}
\DeclareMathOperator*{\argmin}{arg\,min}

\DeclareMathOperator{\diag}{diag}

\newtheorem{theorem}{Theorem}
\newtheorem{lemma}[theorem]{Lemma}

\newtheorem{corollary}[theorem]{Corollary}
\newtheorem{definition}{Definition}